\documentclass{mscs}

\usepackage{amssymb}
\usepackage{stmaryrd}
\usepackage{MnSymbol}
\usepackage{tikz}

\renewcommand{\aleph}{\omega}

\newcommand{\tbSigma}{\mathbf{\Sigma}}
\newcommand{\tbPi}{\mathbf{\Pi}}
\newcommand{\tbDelta}{\mathbf{\Delta}}
\newcommand{\tbD}{\mathbf{D}}
\newcommand{\tlD}{\mbox{D}}
\newcommand{\bGd}{\mathbf{G}_\delta}
\newcommand{\bFs}{\mathbf{F}_\sigma}

\newtheorem{theorem}{Theorem}[section]
\newtheorem{lemma}[theorem]{Lemma}
\newtheorem{definition}[theorem]{Definition}
\newtheorem{example}[theorem]{Example}
\newtheorem{corollary}[theorem]{Corollary}
\newtheorem{proposition}[theorem]{Proposition}
\newtheorem{remark}[theorem]{Remark}
\newtheorem{problem}[theorem]{Open Problem}
\newtheorem{note}[theorem]{Note}

\newcommand{\BM}{\textit{BM}}
\newcommand{\Ch}{\textit{Ch}}

\newcommand{\N}{\mathbb{N}}
\newcommand{\Q}{\mathbb{Q}}
\newcommand{\R}{\mathbb{R}}
\newcommand{\segment}{\!\upharpoonright \!}
\newcommand{\means}[1]{\hbox{$[\kern -.4em [\, {#1}\, ]\kern -.4em]$}}
\newcommand{\conc}{^{\frown}}
\newcommand{\cantor}{\alphabet^\omega}
\newcommand{\alphabet}{{\bf 2}}
\newcommand{\Bool}{{\{0,1\}}}
\newcommand{\mix}{\alphabet^{\leq \omega}}

\newcommand{\Baire}{\N^{\omega}}
\newcommand{\nil}{\textit{nil}}
\newcommand{\rank}{\textit{rank}}

\newcommand{\uup}{\twoheaduparrow\!\!}
\newcommand{\up}{\uparrow\!\!}

\newcommand{\ddown}{\twoheaddownarrow}
\newcommand{\co}{\textit{co-}}
\newcommand{\CK}{\omega_1^{\textit{\scriptsize CK}}}

\newcommand{\PN}{\+P(\N)}
\newcommand{\FIN}{\+P_{\!<\omega}(\N)}

\newcommand{\INF}{\+P_{\!\infty\!}(\N)}

\newcommand{\dom}{\textit{dom}}
\newcommand{\leqpref}{\leq_{\text{pref}}}

\begin{document}

\title[Borel and Hausdorff Hierarchies in Topological Spaces of Choquet Games]{Borel and Hausdorff Hierarchies in Topological Spaces of Choquet Games
\\ and Their Effectivization}

\author[Ver\'onica Becher  and Serge Grigorieff]
{Ver\'onica Becher $^{1}$\thanks{
Members of the Laboratoire International Associ\'e INFINIS,
Universidad de Buenos Aires -- Universit\'e Paris Diderot-Paris 7.
This research was partially done whilst the first author was a visiting fellow
at the Isaac Newton Institute for Mathematical Sciences
in the programme `Semantics \& Syntax'.}
and  Serge Grigorieff $^{2\dagger}$
\\
$^1$ FCEyN,Universidad de Buenos Aires \& CONICET, Argentina.
\addressbreak
$^2$ LIAFA, CNRS \& Universit\'e Paris Diderot - Paris 7, France.}

\date{August 21, 2012}
\maketitle
\bibliographystyle{plain}

\begin{abstract}
What parts of classical descriptive set theory done in Polish spaces
still hold for more general topological spaces,
possibly $T_0$ or $T_1$, but  not $T_2$ (i.e. not Hausdorff)?
This question has been addressed by Victor Selivanov in a series of papers
centered on algebraic domains.
And recently it has been considered by Matthew de Brecht  for  quasi-Polish spaces, a framework that  contains both countably based continuous domains and Polish spaces. 
In this paper we present alternative unifying topological spaces, 
that we call {\em approximation spaces}. 
They are exactly the spaces for which player Nonempty has a stationary strategy in the Choquet game. 
A natural proper subclass of approximation spaces coincides with
the class of quasi-Polish spaces.
We study the Borel and Hausdorff difference hierarchies
in approximation spaces, 
revisiting the work done for the other topological spaces. 
We also consider the problem of effectivization of these results.
\end{abstract}



\section{Introduction}
%
The primary setting of descriptive set theory,
including the study of Borel  and Hausdorff hierarchies,
is that of Polish spaces.
These are spaces  homeomorphic to complete metric spaces
that have a countable dense subset,
for example the Cantor space, the Baire space, the real line and its intervals.
The question of what parts of classical descriptive set theory
still hold for \mbox{non-Polish} spaces,
specifically for  general  $T_0$ topological spaces,
has not been yet systematically studied. 
Major progress has been done by Victor Selivanov in his investigations centered mainly in  algebraic domains
(directed complete partial orders with a countable base of compact elements)
in an ongoing  series of papers on this topic that started in 1978.
Recently,  Matthew de Brecht (2011)  
presented the theory of {\em quasi-Polish} spaces,  
a unifying framework for Polish spaces and countably based domains 
(i.e., $\omega$-continuous domains,
or  directed complete partial orders with a countable basis).
De~Brecht characterized quasi-Polish spaces in terms of the Choquet topological games, and he proved that descriptive set theory can be nicely developed in such spaces.

In this paper we consider alternative unifying topological spaces, 
that we call {\em approximation spaces}.  
Not only they contain all Polish spaces and all continuous domains,
but a natural subclass of approximation spaces coincides with
the class of quasi-Polish spaces.
Approximation spaces can be viewed as the ``\`a la domain" version
of the ``\`a la Polish" unifying framework of de Brecht.
These spaces can also be characterized in terms of Choquet games.
We study the Borel and Hausdorff difference hierarchies in approximation spaces, revisiting the work done for the other topological spaces.
We also consider the problem of effectivization of these results.

The paper is organized as follows.
\S\ref{s:preliminary} presents the preliminary material.  
We recall the needed notions
about the Borel and Hausdorff hierarchies in a $T_0$
(possibly not $T_2$) topological context. 
We give an overview of the needed material on  domains with the Scott topology 
and quasi-Polish spaces.
We also present some prerequisites on the
Banach-Mazur and Choquet topological games.

\S\ref{s:app} is devoted to the class of approximation spaces.
We prove that both, Polish spaces and continuous domains,
 are approximation spaces. 
Indeed, we show that  all quasi-Polish spaces are approximation spaces.
Theorem~\ref{thm:approx games} characterizes  approximation spaces 
in terms of Choquet games.
 Theorem~\ref{thm:coincide} proves that 
quasi-Polish spaces and {\em convergent approximation spaces} are the same class.

In the context of  Polish spaces, the Baire property asserts that
any countable intersection of dense open sets is dense.
Thus, countable intersections of open sets,
the $\bGd$ sets, constitute the $\tbPi^0_2$ level of the Borel hierarchy.
In the context of $T_0$ but not $T_2$ spaces
this is not true any more: the $\tbPi^0_2$ level consists of
countable intersections of {\em Boolean combinations}
of open sets.
Then it is natural to consider the $\tbPi^0_2$ Baire property
which asserts that
any countable intersection of dense differences of open sets
is dense.
As shown by de Brecht (2011),
the usual $\bGd$ Baire property and  Hausdorff-Kuratowski's theorem
both hold for quasi-Polish spaces.
Consequently, these two results are ensured for convergent approximation spaces.
Theorem \ref{thm:Pi2 Baire} proves that, in fact, all approximation spaces
satisfy the $\tbPi^0_2$ Baire property.
Theorem \ref{thm:Hausdorff} extends Hausdorff's theorem
to spaces having a countable basis and such that
every closed subspace is an approximation space:
the $\tbDelta^0_2$ class coincides with the difference hierarchy.
This result was previously obtained  by Selivanov for
$\omega$-algebraic domains and then for $\omega$-continuous domains 
(Selivanov 2005, 2008).
De Brecht (2011) proved that the full  Hausdorff-Kuratowski Theorem
holds for quasi-Polish spaces; hence, it holds for  convergent approximation spaces.
We do not know whether it holds for all approximation spaces.

In \S\ref{s:Hausdorff domains} we revisit a part of Selivanov's work (2005, 2008)
on domains that does not apply to Polish spaces:
his characterization of the classes of the Hausdorff hierarchy in terms of alternating trees,
and his proof of non existence of ambiguous sets in this hierarchy.
We check that the assumption of $\omega$-algebraicity or $\omega$-continuity
can be replaced, mutatis mutandis, by that of continuity.

\S\ref{s:effective hierarchies} is devoted to effectivization.
The definition of approximation spaces admits a straightforward definition of 
an effective version. We make the first steps in developing the effective theory.
We first recall the notions of effective topological space and effective domains.
We also include the known machinery of effective Borel codes.
Thorem~\ref{thm:Hausdorff effective}
proves a weak effective version of Hausdorff's theorem in
effective approximation spaces.
We obtain this proof as an adaptation of Selivanov's work (2003) for the Baire space.
A general effective version of Hausdorff's theorem is still an open question.
\medskip
\\
{\bf Acknowledgement.} We thank Mart\'in Escard\'o
for pointing out  the work of Matthew de Brecht on quasi-Polish spaces
and for giving us very helpful references.
Also we are indebted to an anonymous reviewer for his/her insights
on approximation spaces; 
specifically, we owe this reviewer
the indication to consider Choquet games
and the convergence condition on approximation spaces that 
yields the equivalence with quasi-Polish spaces. 

\section{Preliminary Definitions and Results}
\label{s:preliminary}

We write $\N$ for the set of natural numbers,
$\PN$ for the set of all subsets of $\N$
and $\FIN$ for the set of all finite subsets of $\N$.
Finite sequences of elements of a set $X$ are denoted by
$(x_1,x_2,\ldots, x_n)$.
Concatenation of sequences $u,v$ and element $x$ are written simply as
$uv$, $ux$.
We use Greek letters to denote ordinals. We write 
$\omega$ for the first infinite ordinal, $\omega_1$ for the first uncountable ordinal,
and $\CK$ for the least not computable ordinal (the Church-Kleene ordinal).
For any two ordinals $\alpha,\beta$,  
$\alpha\sim\beta$ means that they have the same parity.

\subsection{Borel and Hausdorff Hierarchies in General Topological Spaces}
\label{s:hierarchies}

All the material of this subsection on the Borel and Hausdorff hierarchies
in general topological spaces has first appeared in \cite{selivanov2005}.
To make the presentation self contained, we reproduce here some of the  proofs.

\subsubsection{The Borel Hierarchy.}
\label{ss:Borel}

In general topological spaces, an open set may possibly not be
a countable union of closed sets, cf. Remark~\ref{rk:PN AH} infra.
In order to get the expected inclusion
$\tbSigma^0_1(E)\subseteq\tbSigma^0_2(E)$,
one has to distort the usual definition of Borel spaces given in
metric spaces.
This leads to define the hierarchy of Borel sets in a general setting as follows.

\begin{definition}[Borel sets] \label{def:Borel}
Let $E$ be a topological space.

\begin{enumerate}
\item  {\it Borel subsets of $E$} are those sets obtained from open sets 
by iterated complementation and countable unions and intersections.

\item The {\it Borel classes}
$\tbSigma_\alpha^0(E)$, $\tbPi_\alpha^0(E)$, $\tbDelta_\alpha^0(E)$,
where $\alpha\geq1$ varies over countable ordinals,
are inductively defined as follows:
$$
\begin{array}{rcl}
\tbSigma^0_1(E)&=&\textit{open subsets of $E$}
\\
\mbox{if $\alpha\geq2$ } \tbSigma^0_{\alpha}(E)
&=&
\textit{countable unions of Boolean combinations
of sets in $\bigcup_{\beta<\alpha}\tbSigma^0_\beta(E)$}
\\
\tbPi^0_\alpha(E)&=&
\{E\setminus X \mid X\in\tbSigma_\alpha^0(E)\}
\\
\tbDelta^0_\alpha(E)&=&
\tbSigma^0_\alpha(E) \cap \tbPi^0_\alpha(E)
\end{array}
$$

\item  The class $\bGd(E)$ (respectively $\bFs(E)$) is the family of
countable intersections of open sets (respectively unions of closed sets).
In general, it is a proper subclass of $\tbPi^0_2(E)$
(respectively $\tbSigma^0_2(E)$).

\end{enumerate}
\end{definition}
The following result follows from elementary set theory.

\begin{proposition}\label{p:Borel}
\begin{enumerate}
\item  $\tbSigma^0_\beta(E)\cup\tbPi^0_\beta(E)
         \subseteq    \tbDelta^0_\alpha(E)$  for any $\alpha>\beta\geq1$.
\item Each one of the Borel classes $\tbSigma^0_\alpha(E)$,
$\tbPi^0_\alpha(E)$, $\tbDelta^0_\alpha(E)$
is closed under finite unions and intersections
and continuous inverse images.
The $\tbSigma^0_\alpha(E)$ (respectively $\tbPi^0_\alpha(E)$) classes
are closed under countable unions (respectively intersections).
\end{enumerate}
\end{proposition}
As expected, the above definition is equivalent to the
usual one for metric spaces.
Also, in the general case, the distortion can be done solely
for $\tbSigma^0_2(E)$.

\begin{proposition}\label{p:Borel normal}
\begin{enumerate}
\item  $\tbSigma^0_2(E)$ coincides with the family of countable unions
of differences of sets in $\tbSigma^0_1(E)$,
i.e. sets of the form $\bigcup_{n\in\N} U_n\setminus V_n$
where the $U_n, V_n$'s are open.
Moreover, if $\+B$ is a countable topological basis then one can take
the $U_n$'s in $\+B$.
\item If $\alpha\geq3$
then $\tbSigma^0_\alpha(E)$ is the family of countable unions
of sets in $\bigcup_{\beta<\alpha}\tbPi^0_\beta(E)$.
If $E$ is metrizable then this also holds for $\alpha=2$.
\end{enumerate}
\end{proposition}

\begin{proof}
1. Observe that a Boolean combination of open sets is a finite union of
differences of two open sets.
For the last assertion, use  that $U_n$ is a union of sets in~$\+B$.

2. It suffices to prove that the difference $X\setminus Y$
of two sets in $\tbSigma^0_\beta(E)$, with $\beta<\alpha$,
is equal to a countable union of sets in $\tbPi^0_\gamma(E)$
with $\gamma<\alpha$.
In case $\beta+1<\alpha$ then, as the intersection of a
$\tbSigma^0_\beta(E)$ and a $\tbPi^0_\beta(E)$ set,
$X\setminus Y$ is $\tbPi^0_{\beta+1}(E)$ and we are done.
In case $\alpha=\beta+1$, since $\alpha\geq3$, we have $\beta\geq2$
and $X$ is of the form $X=\bigcup_{i\in\N}U_i\setminus V_i$
where $U_i,V_i$ are in $\bigcup_{\gamma<\beta}\tbSigma^0_\gamma(E)$.
Thus,
$X\setminus Y=\bigcup_{i\in\N}(U_i\setminus V_i)\setminus Y
=\bigcup_{i\in\N}U_i\cap(E\setminus(V_i\cup Y))$.
Now, $U_i\in\tbSigma^0_\gamma(E)$, with $\gamma<\beta$,
hence $U_i\in\tbPi^0_{\gamma+1}(E)$ where $\gamma+1\leq\beta$.
Also, $V_i\cup Y\in\tbSigma^0_\beta(E)$ hence
$E\setminus(V_i\cup Y)\in \tbPi^0_\beta(E)$.
Therefore, $X\setminus Y$ is a countable union of sets in $\tbPi^0_\beta(E)$.
Finally, in a metric space, the topological closure $\overline{X}$
of any set $X$ is $\bGd$ since
$\overline{X}=\bigcap_{n\in\N}
                   \{z\mid \exists x\in X\ d(z,x)<2^{-n}\}$.
Going to complements, any open set is $\bFs$.
Then, any difference of two open sets
hence also any $\tbSigma^0_2(E)$ set is also $\bFs$,
i.e. a countable union of $\tbPi^0_1(E)$ sets.
\end{proof}

\subsubsection{The Hausdorff Difference Hierarchy.}
Recall the Hausdorff difference infinitary operation,
cf. \cite{kuratowskiBook,kechrisBook}.

\begin{definition}
\label{def:D}
Let $\alpha$ be an ordinal.
\begin{enumerate}
\item
The {\it difference operation} $D_\alpha$
maps an $\alpha$-sequence of subsets $(A_\beta)_{\beta<\alpha}$
of a space $E   $ to the subset
$\displaystyle{D_\alpha((A_\beta)_{\beta< \alpha}) =
\bigcup_{\beta< \alpha,\ \beta\not\sim\alpha}
A_\beta \setminus \cup_{\gamma<\beta} A_\gamma}$
\item
We let
$\co D_\alpha((A_\beta)_{\beta< \alpha})
=E\setminus D_\alpha((A_\beta)_{\beta< \alpha})$.
\item
For a class of subsets $\+A$,
we let $\tbD_\alpha(\+A)$ (respectively $\co\tbD_\alpha(\+A)$)
be the class of all subsets
$D_\alpha((A_\beta)_{\beta< \alpha})$
(respectively $\co D_\alpha((A_\beta)_{\beta< \alpha})$),
where $A_\beta\in\+A$ for all $\beta< \alpha $.
\end{enumerate}
\end{definition}

\begin{remark}\label{rk:D2}
In particular, $\tbD_2(\+A)$ (respectively $\co\tbD_2(\+A)$)
is the family of sets
$A_1\setminus A_0$ (respectively $A_0\cup(E\setminus A_1)$)
with $A_0,A_1\in\+A$.
\end{remark}

\begin{proposition}\label{p:Hausdorff}
\begin{enumerate}
\item  If $\emptyset\in\+A$
then $\tbD_\beta(\+A) \subseteq \tbD_\alpha(\+A)$
for all $\beta<\alpha$.
\item 
 If $E\in\+A$
then $\co\tbD_\alpha(\+A) \subseteq \tbD_{\alpha+1}(\+A)$.
In particular, if $\emptyset,E\in\+A$ and $\beta<\alpha$
then
$\tbD_\beta(\+A)\cup\co\tbD_\beta(\+A)
\subseteq \tbD_\alpha(\+A)\cap\co\tbD_\alpha(\+A)$.
\item If $\+A$ is closed under countable unions then,
for $\alpha$ countable, in the definition of $\tbD_\alpha(\+A)$,
one can restrict to monotone increasing $\alpha$-sequences.
\end{enumerate}
\end{proposition}

\begin{proof}
1. If $\beta\sim\alpha$ then
$D_\beta((A_\gamma)_{\gamma< \beta})
=D_\alpha((A'_\delta)_{\delta< \alpha})$
where $A'_\delta=A_\delta$ for $\delta<\beta$
and $A'_\delta=\emptyset$ for $\delta\geq\beta$.
If $\beta\not\sim\alpha$ then
$D_\beta((A_\gamma)_{\gamma< \beta})
=D_\alpha((A'_\delta)_{\delta< \alpha})$
where $A'_{\delta+1}=A_\delta$ for $\delta<\beta$
and $A'_\delta=\emptyset$ for $\delta=0$
or $\delta$ limit or $\delta\geq\beta$.

2. Observe that
$D_\alpha((A_\beta)_{\beta< \alpha})
=D_\alpha((A'_\beta)_{\beta< \alpha})$
where $A'_\beta=\bigcup_{\gamma\leq\beta}A_\gamma$.

3. Letting $A_\alpha=E$, we have
$\co D_\alpha((A_\beta)_{\beta< \alpha})
=D_{\alpha+1}((A_\beta)_{\beta\leq\alpha})$.
\end{proof}

\begin{definition}
For any $0<\beta<\aleph_1$,
the $\alpha$-th level of the {\it difference hierarchy over
$\tbSigma^0_\beta(E)$} is
$\tbD_\alpha(\tbSigma^0_\beta(E))$.
The difference hierarchy over $\tbSigma^0_1(E)$
is simply called the difference hierarchy and denoted  by
$\tbD_\alpha(E)$.
\end{definition}

\begin{remark}
\begin{enumerate}
\item  Using item 2 of Proposition~\ref{p:Hausdorff},
we can graphically represent sets in the first levels
of the difference hierarchy as in Figure~\ref{fig:D0 D4}.

\item
 This graphical representation makes it clear that
$\tbD_\alpha(E)$ is not closed under finite union
nor finite intersection:
for instance, if $A_0\subset A_1\subset A_2$ then
\\$\begin{array}{lccc}
&D_2(\emptyset,A_0)\cup D_2(A_1,A_2)&=&D_3(A_0,A_1,A_2)\\
\textit{hence\qquad}&
\co D_2(\emptyset,A_0)\cap\co D_2(A_1,A_2)
&=&\co D_3(A_0,A_1,A_2)\phantom{\co}\\
\textit{and\qquad}&
D_3(\emptyset,A_0,E)\cap D_3(A_1,A_2,E)
&=&\phantom{\co\ .}D_4(A_0,A_1,A_2,E)\ .
\end{array}$
\end{enumerate}
\end{remark}

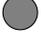
\begin{figure}
\vspace{18mm}
\begin{tikzpicture}
\tikz
\draw[fill=gray] (0,0) circle(2.3mm);
\hspace{5mm}
\tikz
\draw[line width=4.6mm,color=gray] (0,0) circle(4.6mm);
\hspace{5mm}
\tikz
\draw[line width=4.6mm,color=gray] (0,0) circle(9.2mm);
\draw[fill=gray] (-1.15,11.5mm) circle(2.3mm); 
\hspace{5mm}
\tikz
\draw[line width=4.6mm,color=gray] (0,0)
circle(4.6mm) circle(13.8mm);
\end{tikzpicture}
\vspace{-1.5cm}
\caption{In grey: $D_1(A_0)$,
$D_2(A_0,A_1)$,
$D_3(A_0,A_1,A_2)$,
$D_4(A_0,A_1,A_2,A_3)$
where $A_0\subset A_1\subset A_2\subset A_3$.
In white (including the unbounded complement of the largest disk):
$\co D_1(A_0)$,
$\co D_2(A_0,A_1)$,
$\co D_3(A_0,A_1,A_2)$,
$\co D_4(A_0,A_1,A_2,A_3)$.}
\label{fig:D0 D4}
\end{figure}

\begin{proposition}\label{p:Hausdorff Delta Borel}
$\bigcup_{\alpha<\aleph_1}\tbD_\alpha(\tbSigma^0_\beta(E))
\subseteq \tbDelta^0_{\beta+1}(E)$.
\end{proposition}

\begin{proof}
If the $A_\gamma$'s, $\gamma<\alpha$ are in $\tbSigma^0_\beta(E)$
then so are the $\bigcup_{\delta<\gamma}A_\delta$'s.
Thus, $D_\alpha((A_\gamma)_{\gamma<\alpha})$ is a countable
union of differences of sets in $\tbSigma^0_\beta(E)$
hence is in $\tbSigma^0_{\beta+1}(E)$.
By Proposition~\ref{p:Hausdorff},
we see that
$E\setminus D_\alpha((A_\gamma)_{\gamma<\alpha})
=D_{\alpha+1}((A_\gamma)_{\gamma<\alpha},E) \in \tbPi^0_{\beta+1}(E)$.
\end{proof}

\begin{proposition}\label{p:inverseimagesHausdorff}
Let $D,E$ be topological spaces and $f:D\to E$ be continuous.
If $Y\subseteq E$ is in some Hausdorff class $\tlD_\alpha(E)$,
$\alpha<\aleph_1$,
then $f^{-1}(\+Y)$ is in $\tlD_\alpha(D)$.
\end{proposition}

\subsubsection{The Borel and Hausdorff Hierarchies May Collapse.}
%
As it is well-known, the Borel and Hausdorff hierarchies are proper
in uncountable Polish spaces:
$\tbSigma^0_\alpha(E) \subsetneq \tbSigma^0_\beta(E)$
and
$\tbD_\alpha(\tbSigma^0_\xi(E))
\subsetneq \tbD_\beta(\tbSigma^0_\xi(E))$
when $\alpha<\beta$.
The same for the effective hierarchies
(with $\alpha,\beta<\CK$)
relative to some fixed enumeration of a countable basis of open sets.
However, this is not true in general topological spaces.
For instance, if the space is $T_2$ and countable then
$\tbSigma^0_2(E)=\+P(E)$.
However, \cite{deBrecht2011} proves the non collapse of
the Borel and Hausdorff difference hierarchies
in uncountable quasi-Polish spaces,
a class containing Polish spaces and $\omega$-continuous domains.

\subsection{Domains}\label{s:domains}

Domain theory refers to the field initiated  by Dana Scott in the late 1960s
to specify denotational semantics for functional programming languages.
The theory formalizes the ideas of approximation and convergence
via some partially ordered sets called {\em domains}.

\begin{example}\label{ex:Scott examples}
Some examples of Scott topologies on partially ordered sets.
\begin{enumerate}
\item {\bf Scott topology on $(\PN,\subseteq)$.}
For $A\in\FIN$, let
$O_A=\{X\in\PN\  \mid\  X\supseteq A\}$.
The Scott topology on $\PN$ is has the $O_A$'s, for $A\in\FIN$, as a topological basis.
This is the topology of ``positive information'';
in contrast, the Cantor topology on $\cantor$ gives positive
and negative information.

\item {\bf Scott topology on $(\INF,\subseteq)$ :}
consider the family $\INF$ of infinite subsets of $\N$
as a topological subspace of $\PN$.

\item
 {\bf Scott topology on the family $(X^{\leq\omega},\leqpref)$ of finite or infinite $X$-sequences.}
We suppose $X$ is any set with at least two elementss.
For any $s\in X^{<\omega}$, let
$\+B_s=\{u\in X^{\leq\omega} \mid u\textit{ extends }s\}$.
The Scott topology on the space $X^{\leq\omega}$ is that which admits
the $\+B_s$'s, $s\in X^{<\omega}$, as a topological basis.

\item
 {\bf Scott topology on the right extended real line $(\overrightarrow{\R},\leq)$.}
Let $\overrightarrow{\R}=\R\cup\{+\infty\}$.
The open sets of the topology on $\overrightarrow{\R}$ are 
$\overrightarrow{\R}$ and the semi-intervals $]x,+\infty]$, for $x\in\R$.

\item {\bf Extended real line $(\widetilde{\R},\leq)$ with duplicated rationals.}
Let $\widetilde{\R}=\R\cup(\Q\times\{+\})\cup\{+\infty\}$
and $\leq$ be the following total order: 
\ $+\infty$ is a maximum element, and for all $x,y\in\R$ and $q\in\Q$,
\ $q<_{\widetilde{\R}}(q,+)$;
\  $x<_{\widetilde{\R}}y$ if and only if $x<y$;
and 
\ $x<_{\widetilde{\R}}(q,+)<_{\widetilde{\R}}y$
if and only if $x<q<y$.
The Scott topology on $\widetilde{\R}$ is that for which the $[(q,+),+\infty]$'s,
$q\in\Q$, are a topological basis.

\end{enumerate}
\end{example}
The following properties are straightforward.

\begin{proposition}\label{p:PN T0}
The Scott topologies of spaces in Examples~\ref{ex:Scott examples}
are not $T_2$  
but are $T_0$,
i.e. they satisfy Kolmogorov's axiom:
given two distinct points, one of them has a neighborhood
which does not contain the other one
(but this may not be symmetric).
\end{proposition}

\begin{remark}\label{rk:PN AH}
\begin{enumerate}
\item  As noticed in \cite{selivanov2005},
the finite levels of the Scott Borel
hierarchy on $\PN$ do not coincide
with the corresponding ones on the Cantor space~$\cantor$~:
$\tbSigma^0_n(\PN)
\subsetneq \tbSigma^0_n(\cantor)
\subsetneq \tbSigma^0_{n+1}(\PN)$
for all~$n\in\N$.
The same is true with the effective Borel hierarchy
(cf. \S~\ref{ss:Borel codes}).
For instance,
$\+X=\PN\setminus\{\N\}$, defined by the formula
$\exists x\ (x\notin X)$, is $\Sigma^0_1(\cantor)$ and
$\Sigma^0_2(\PN)$ but neither Scott open nor Scott closed.
However, the infinite levels of the Borel hierarchy 
on $\PN$ and $\cantor$ coincide.

\item  The only subsets of $\PN$ that are both open and $F_\sigma$
are $\emptyset$ and $\PN$.
Indeed, suppose $\+O$ is open and $\+X$ is $F_\sigma$
and $\+O,\+X$ are different from $\emptyset,\PN$.
Then there exist non empty finite subsets $X,Y$ of $\N$ such that
$O_X\subseteq\+O$ and $\PN\setminus O_Y\subseteq\+X$.
Observe that the set $X\cup Y$ is in $\+O\setminus\+X$,
showing $\+O\neq\+X$.
\end{enumerate}
\end{remark}

\subsubsection{The Scott Topology on dcpo's.}

We briefly recall the main definitions and notions
and refer the reader to classical papers and books,
for instance~\cite{abramsky1994,edalat1997,BookDomains}.

\begin{definition}\label{def:dcpo}
\begin{enumerate}
\item 
 A {\it directed complete partial order (dcpo)} is a partially ordered set
$(D,\sqsubseteq)$ such that every non empty directed subset $S$ has a
least upper bound (denoted by $\sqcup S$).
A dcpo is pointed if it has a least element $\bot$.

\item  The {\it Scott topology} on a dcpo is the topology that admits as closed sets
all sets $X$ satisfying conditions 
\begin{itemize}
\item
$X$ is a downset: $x\in X \wedge y\sqsubseteq x \Rightarrow y\in X$.
\item
$X$ is closed under suprema of directed subsets of $D$.
\end{itemize}
Then, $O\subseteq D$ is open in the Scott topology if it satisfies
the following conditions.
\begin{itemize}
\item
$O$ is an upset: $x\in O \wedge x\sqsubseteq y \Rightarrow y\in O$.
\item
Every directed set with supremum in $O$ has an element in $O$.
\end{itemize}
\end{enumerate}
\end{definition}

\begin{example}\label{ex:Scottopen}
For every $x\in D$, the set $U_x=\{z \mid z\not\sqsubseteq x\}$ is Scott open.
\end{example}

\begin{proposition}\label{p:Scotttopology}
\begin{enumerate}
\item The Scott topology on a dcpo is $T_0$, i.e.
if $x\neq y$ then there exists an open set which contains
only one of the two points $x,y$.
It is $T_1$ (respectively $T_2$ 
) if and only if the order on $D$
is trivial.
\item Let $x,y\in D$.
The order on $D$ can be recovered from the topology
as the specialization order:
$x\sqsubseteq y$ if and only if every Scott open set containing $x$ also contains $y$
if and only $x\in \overline{\{y\}}$
(where $\overline{\{y\}}$ is the topological closure of $\{y\}$).
\item
 A function $f:D\to E$ between two dcpo's is continuous with respect
to the Scott topologies if and only if it is monotone increasing
and preserves suprema of directed subsets: if $S\subseteq D$ is directed
then $f(\sqcup S)=\sqcup f(S)$.
\end{enumerate}
\end{proposition}

\subsubsection{The Scott Topology on Domains.}

\begin{definition}\label{def:way below}
\begin{enumerate}
\item Let $(D,\sqsubseteq)$ be a dcpo.
The {\it approximation (or way-below) relation} on $D$ is defined as follows:
let $x,y\in D$, $x\ll y$ if, for all directed subset $S$,
$y \sqsubseteq \sqcup S$ implies $x\sqsubseteq s$ for some $s\in S$.
We say $x$ approximates, or is way-below, $y$.

\item  An element $x\in D$ is {\it compact (or finite)} if $x\ll x$.
The set of compact elements is denoted by $K(D)$.

\item  $\uup x= \{y \mid \ x\ll y\}$ and $\ddown x=  \{y \mid \ y \ll x\}$
\end{enumerate}
\end{definition}

\begin{proposition}\label{p:waybelow}
Let $(D,\sqsubseteq)$ be a dcpo and $x,x',y,y'\in D$. 

\begin{enumerate}
\item 
$x\ll y\ \Rightarrow\ x\sqsubseteq y$, and 
$(x'\sqsubseteq x\ll y\sqsubseteq y')\ \Rightarrow\ x'\ll y'$.
\\
If $x$ is compact then
$\forall u,v\ ((u\sqsubseteq x\sqsubseteq v)\ \Leftrightarrow\ ( u\ll x\ll v))$.
\item 
 An element $x$ is compact if and only if 
$\up x=\{y\mid x\sqsubseteq y\}$ is Scott open.
\end{enumerate}
\end{proposition}

\begin{proposition}\label{p:domain}
Let $(D,\sqsubseteq)$ be a dcpo.
The following conditions are equivalent.
\begin{enumerate}[(iii)]
\renewcommand{\theenumi}{(\roman{enumi})}
\item
For every $x\in D$, the set $\ddown x =\{z\in D \mid z\ll x\}$
is directed and $x=\sqcup \ddown x$.
\item
There exists $B\subseteq D$ such that,
for every $x\in D$, $B\cap \ddown x$
is directed and $x=\sqcup(B\cap \ddown x)$.
\end{enumerate}
$D$ is a {\it continuous domain} if these conditions hold.
Any set $B$ satisfying condition (ii) is called a basis.
$D$ is an {\it $\omega$-continuous domain} if (ii) holds for some
countable set $B$.
$D$ is an algebraic (respectively $\omega$-algebraic) domain
if $K(D)$ is a basis (and is countable).
\end{proposition}

\begin{example}\label{ex:domain}
All spaces in Example~\ref{ex:Scott examples} are dcpo's with the Scott topology.
The spaces $\PN$, $X^{\leq\omega}$, $\widetilde{\R}$
are $\omega$-algebraic domains.
Their sets of compact elements are
respectively $\FIN$, $X^{<\omega}$ and $\Q\times\{+\}$.
The way-below relation on any of these three spaces is the restriction
of the partial order $\leq$ to $K(D)\times D$.
The space $\overrightarrow{\R}$ is an $\omega$-continuous domain but
is not algebraic: there is no compact element
and its way-below relation is the strict order $<$.
The space $\INF$ is not continuous: its way-below relation is empty.
\end{example}
Let us recall classical results in continuous domains.

\begin{proposition}\label{p:Scott topology}
Let $(D,\sqsubseteq)$ be a continuous domain with basis $B$.
\begin{enumerate}
\item  {\bf (Interpolation property).}
If $M\subseteq D$ is finite and $a\ll x$ for each $a\in M$
then there exists $x'\in B$ such that $M\ll x'\ll x$.
\item
 $x\ll y$ if and only if $y$ is interior to the upper cone
$\up x=\{z\mid   x\sqsubseteq z \}$.
\item 
 A set $O$ is open
if and only if $O=\bigcup_{x\in O} \uup x$
if and only if $O=\bigcup_{x\in O\cap B} \uup x$.
In particular, the family of sets $\uup z$,
where $z$ varies in $B$, is a basis of the Scott topology on $D$.
\end{enumerate}
\end{proposition}

\subsection{Quasi-Polish Spaces}
\label{ss:quasi-Polish}

Quasi-Polish spaces, 
developed by de Brecht (2011), 
are a unifying framework of Polish spaces and
$\omega$-continuous domains. 
We recall here the definition and main results. 

\begin{definition}\label{def:quasi}
\begin{enumerate}
\item  Giving up the symmetry axiom of metrics,
a quasi-metric on a space $E$ is defined as a function
$d:E^2\to[0,+\infty[$
such that, for all $x,y,z\in E$,
$$
x=y\Leftrightarrow d(x,y)=d(y,x)=0
\quad,\quad
d(x,z)\leq d(x,y)+d(y,z)\ .
$$
The topology associated to $d$ is the one generated by the open balls.
\item
 A sequence $(x_n)_{n\in\N}$ is Cauchy if
$\lim_{n\to+\infty}\sup_{p\geq n}d(x_n,x_p)=0$.
\item
 A quasi-metric space $(E,d)$ is complete if every Cauchy
sequence is convergent relative to the metric $\widehat{d}$
such that $\widehat{d}(x,y)=\max\{d(x,y),d(y,x)\}$.
\item Topological spaces associated to
complete quasi-metrics with a countable topological basis
are called quasi-Polish.
\end{enumerate}
\end{definition}

\begin{example}\label{ex:quasiPolish}
The Scott topology on $\PN$ is quasi-Polish for the quasi-metric
such that $d(X,Y)=\sum_{n\in X\setminus Y}2^{-n}$.
\end{example}

\begin{theorem}\cite{kunzi1983}\label{thm:kunzi}
A quasi-metric space $(E,d)$ has a countable topological basis
if and only if the metric space $(E,\widehat{d})$ is separable
(i.e. has a countable dense subset).
\end{theorem}
The following theorem sums up some of the main results
in \cite{deBrecht2011}.

\begin{theorem}\cite{deBrecht2011}\label{thm:deBrecht}.
\begin{enumerate}
\item  A space is quasi-Polish if and only if it is homeomorphic to some
$\tbPi^0_2$ subspace of $\PN$ endowed with the Scott topology.

\item  Polish spaces and $\omega$-continuous domains are quasi-Polish.
\item  Every quasi-Polish space $E$ satisfies the following properties:
\\
- (Baire property)\  The intersection of a sequence
of dense open sets is a dense set.
\\
- (Hausdorff-Kuratowski property)\
$\bigcup_{\alpha<\aleph_1}\tbD_\alpha(\tbSigma^0_\beta(E))
= \tbDelta^0_{\beta+1}(E)$ for all
$0<\beta<\aleph_1$.
\end{enumerate}
\end{theorem}

\begin{remark}
Two examples of subspaces of $\PN$ illustrate the above theorem.
\begin{enumerate}
\item
The subspace
$C=\{X\in\PN\mid\forall i\in\N\ (2i\in X\Leftrightarrow 2i+1\notin X)\}
=\bigcap_{i\in\N}(O_{2i}\Delta O_{2i+1})$
is $\tbPi^0_2$ in $\PN$ and is homeomorphic to the Cantor space
$\cantor$ (and is therefore Polish).
\item
The subspace $\INF=\bigcap_{i\in\N}\bigcup_{j\geq i}O_{\{j\}}$
is $\tbPi^0_2$ in $\PN$ hence is quasi-Polish.
It has a countable basis but is not Polish since it is $T_0$ and not $T_2$.
Although $(\INF,\subseteq)$ is a dcpo, it is not a continuous domain
since its way-below relation is empty.
\end{enumerate}
\end{remark}

\subsection{Topological Games}
\label{ss:games}

The Choquet topological games have been used to characterize
Polish spaces \cite{choquet1969} and quasi-Polish spaces
\cite{deBrecht2011}. We shall use them to characterize
approximation spaces (cf. Definition~\S\ref{def:approximation} infra).
Choquet games are a variant of Banach-Mazur games.
The subtlety of this variant is best understood by confronting
definitions and properties of both classes of games.
Also, the interest of de Brecht's Theorem 2.36 and our Corollary 3.12 about
convergent Markov/stationary strategies in the Choquet game
is highlighted by the counterpart (and more powerful) results
by Galvin \& Telg\' arsky for Banach-Mazur games
(cf. Corollary 2.34, Remark 2.35).

\subsubsection{Banach-Mazur and Choquet Games.}
\label{sss:games}
Let us recall the classical definitions.

\begin{definition}\cite{choquet1969,galvin1986}
Let $X$ be a topological space.
\begin{enumerate}
\item In the Banach-Mazur game $\BM(X)$ two players,
Empty and Nonempty, alternate turns for $\omega$ rounds.
On round $0$ (respectively $i+1$), Empty moves first,
choosing a non empty open subset $U_0\subseteq X$
(respectively $U_{i+1}\subseteq V_i$).
Then, Nonempty responds with a non-empty open set
$V_0\subseteq U_0$ (respectively $V_{i+1}\subseteq U_{i+1}$).
After all the rounds have been played,
Nonempty wins if $\bigcap_{i\in\N} V_i \neq \emptyset$.
Otherwise, Empty wins.

\item The Choquet game $\Ch(X)$ is the variant of the Banach-Mazur game
$BM(X)$ where at round $i$ Empty picks a pair $(x_i,U_i)$
such that $U_i$ is open and $x_i\in U_i$,
and Nonempty picks an open set $V_i$
such that $x_i\in V_i\subseteq U_i$.

\item  A winning strategy for a player (in any of the above games)
is a function that takes a partial play of the game
ending with a move by its opponent
and returns a move to play,
such that the player wins any play of the game that follows the strategy.

\item A winning strategy for Nonempty is convergent if, when he follows it,
the $V_i$'s are a basis of neighborhoods of some
$x\in\bigcap_{i\in\N}V_i$.
\end{enumerate}
\end{definition}

\begin{note}\label{note:games}
The denominations ``Banach-Mazur" and ``Choquet"
are the most commonly used (and are historically accurate).
However, in \cite{kechrisBook}  these games are respectively called
``Choquet" and ``strong Choquet".
\end{note}

\begin{remark}\label{rk:ChversusBM}
Every winning strategy for Nonempty in the Choquet game
$\Ch(X)$ yields one in the Banach-Mazur game $\BM(X)$.
\end{remark}
The following theorem sums up some known results
around these topological games.

\begin{theorem}\label{thm:games}
Let $X$ be a topological space.
\begin{enumerate}
\item  (Oxtoby 1957, cf. Kechris 1995, Theorem 8.11)
$X$ has the Baire property
(i.e. the intersection of countably many dense open sets is dense)
if and only if
player Empty has no winning strategy in the Banach-Mazur game $\BM(X)$.

\item  (Choquet 1969, cf. Kechris 1995, Theorem 8.18)
$X$ is Polish if and only if $X$ has a countable basis,
is $T_1$ and regular and
player Nonempty has a winning strategy
in the Choquet game $\Ch(X)$.
\item  (de Brecht 2011, Theorem 51)
$X$ is quasi-Polish if and only if $X$ has a countable basis,
and player Nonempty has a convergent winning strategy in the Choquet game $\Ch(X)$.
\end{enumerate}
\end{theorem}

\subsubsection{Markov and Stationary Strategies.}
\label{sss:stationary}
The pioneer work of Schmidt (1966) and Choquet (1969) considered strategies of a very simple form.  Then Galvin \& Telg\'arky (1986) obtained deep results for other strategies.

\begin{definition}\label{def:Markov stationary}
A winning strategy is stationary (respectively Markov)
if it depends only on the last move of the opponent
(respectively and on the ordinal rank of the round).
\end{definition}

\begin{theorem}\label{thm:Markov stationary}
(Galvin \& Telg\'arky 1986, Theorems 5 and 7)
Let $(S, \leq)$ be a non empty partially ordered set.
Let $R$ be a family of monotone (non strictly) decreasing sequences in $S$.
In the game $G(S,\leq,R)$,
player I starts and plays elements $a_0,a_1,\ldots$ of $S$,
player II plays elements $b_0,b_1,\ldots$ of $S$
in such a way that
$a_{i+1}\leq b_i$ and $b_i\leq a_i$ for all $i$.
Player II wins if and only if $(a_0,a_1,\ldots)\in R$.
Suppose $R$ is such that
for all  monotone (non strictly) decreasing sequences
$(x_i)_{i\in\N}$, $(y_i)_{i\in\N}$  in $S$,
$$
(\forall i\ \exists j\ y_j\leq x_i)
\wedge(\forall j\ \exists k\ x_k\leq y_j)
\wedge 
(x_i)_{i\in\N}
\in R
\ \Rightarrow\ 
(y_i)_{i\in\N} \in R\ .
$$
\begin{enumerate}
\item  If player II has a winning strategy in $G(S,\leq,R)$
then he has one that depends only on the last move of I
and the last move of II.
\item  If player II has a Markov winning strategy in $G(S,\leq,R)$
then he has a stationary one.
\end{enumerate}
\end{theorem}

\begin{corollary}\label{coro:Markov stationary}
(Galvin \& Telg\'arky 1986, Corollaries 9 and  14)
\begin{enumerate}

\item If player Nonempty has a winning (respectively and convergent)
strategy in the Banach-Mazur game $\BM(X)$
then he has one which depends only on the last move of Empty
and the last move of Nonempty (respectively and is convergent).

\item  If player Nonempty has a Markov (respectively and convergent)
winning strategy in $\BM(X)$
then he has a stationary (respectively and convergent) one.

\item  (Debs 1984, 1985)
There exists a $T_2$ 
(even completely regular) space $X$ such that
player Nonempty has no stationary winning strategy in $\BM(X)$
but has a winning strategy 
which depends only on the last two moves of Empty.
\end{enumerate}
\end{corollary}

\begin{remark}
Theorem~\ref{thm:Markov stationary} does not apply to
the Choquet game because the players do not play
in the same partially ordered set.
For simple winning strategies for Nonempty in the Choquet game $\Ch(X)$
in particular spaces $X$, see \S5 of \cite{bennetlutzer2009}.
\end{remark}
The proof  of item 3 of Theorem~\ref{thm:games}  
given by de Brecht (Theorem 51 in de Brecht 2011) 
yields more:

\begin{theorem}\label{thm:quasiPolish strategies}
\cite{deBrecht2011}
If $X$ is quasi-Polish then player Nonempty has a Markov
convergent winning strategy in the Choquet game $\Ch(X)$.
\end{theorem}
We shall improve this last result (replacing Markov by stationary),
cf. Corollary~\ref{coro:Choquet games convergent}.

\section{Approximation Spaces: the Spaces of Choquet Games}
\label{s:app}

\subsection{Approximation Spaces}
\label{ss:approximation spaces}

We introduce another class of topological spaces:
approximation spaces.
They include all continuous domains, all Polish spaces 
and, in fact, all quasi-Polish spaces. 
The definition is based on an {\it approximation relation}
 which formalizes a containment relation between basic open sets.
This containment relation ensures that inclusion-decreasing chains have a non empty intersection; however,  this intersection may not be reduced to a singleton set.
An  example of an approximation relation 
is obtained by lifting to basic open sets the way-below relation in a dcpo.
We borrowed the notation $\ll$ from this particular example.

Theorem \ref{thm:coincide} proves that a subclass of  second-countable
approximation spaces coincides with the class of quasi-Polish spaces. 
This gives an ``\`a la domain" characterization of quasi-Polish spaces.
Whether the notion of approximation space captures a substantial part
of the rich theory developed by de Brecht for quasi-Polish spaces
is a question still to be investigated.
In the next sections, we just  show two pleasant properties:
a $\tbPi^0_2$ Baire property and Hausdorff's characterization
of $\tbDelta^0_2$.

\begin{definition}\label{def:approximation}
Let $E$ be a topological space.

\begin{itemize}

\item 
An {\it approximation relation} for $E$ is a binary relation
$\ll$ on some topological basis $\+B$ such that, for all $U,V,T\in\+B$,
\begin{enumerate}
\item[(1)]\
If $U\ll V$ then $V\subseteq U$,
\item[(2)]\
If $U\subseteq T$ and $U\ll V$ then $T\ll V$
(in particular, $\ll$ is transitive),
\item[(3)]\
for all $x\in U$, there exists $W\in\+B$ such that $x\in W$ and $U\ll W$,
\item[(4)]\
For every sequence ($U_i)_{i\in\N}$ of sets in $\+B$
such that $U_i \ll U_{i+1}$ for all $i$,
the intersection set $\bigcap_{i\in\N}U_i$ is non empty.
\end{enumerate}

\item {\em Convergent approximation relations}
are obtained by strengthening condition (4) to
\begin{enumerate}
\item[(4${}^+$)]\ \
Every sequence ($U_i)_{i\in\N}$ of sets in $\+B$
such that $U_i \ll U_{i+1}$ for all $i$,
is a neighborhood basis of some $x\in\bigcap_{i\in\N}U_i$
(i.e.  each open set containing $x$ also contains some $U_i$).
\end{enumerate}

\item {\em An approximation relation is  hereditary} if,
for every closed subset $C$ of $E$,
$$\ll_C\ \ =\ \ \{(C\cap U,C\cap V) \mid
              U\ll V\textit{ and }C\cap U,C\cap V\neq\emptyset\}$$
is an approximation relation for the subspace $C$.              

\item A space $E$ is an {\em  approximation space}
(respectively {\em convergent approximation space},
respectively {\em hereditary approximation space})
if it admits an approximation
(respectively convergent approximation, respectively hereditary approximation) relation.
\end{itemize}
\end{definition}

\begin{example}\label{ex:trivial approx}
Every subspace $D$ of the Scott domain $\PN$ which is an upset
(i.e. if $X\subseteq Y$ and $X\in D$ then  $Y\in D$)
is a trivial approximation space:
containment is an approximation relation on the basis
$\{D\cap O_A\mid\alpha\in\FIN\}$ where $O_A=\{X\mid A\subseteq X\}$
since $\N$ belongs to all $D\cap O_A$'s.
In particular, the dcpo $(\INF,\subseteq)$ is an approximation space
which is neither Polish nor a continuous domain.
In $\PN$, inclusion is a
convergent approximation relation since
$\bigcup_{i\in\N}A_i\in\bigcap_{i\in\N}O_{A_i}$
and the $O_{A_i}$'s converge to $\bigcup_{i\in\N}A_i$.
\end{example}
Approximation relations exist on all topological basis or on none.

\begin{lemma}\label{l:approximation and basis}
Let $\+B$ and $\+C$ be topological basis of a space $E$.
If there exists an
(respectively  convergent; respectively  hereditary) 
approximation relation on $\+B$ 
then there exists one on $\+C$.
\end{lemma}

\begin{proof}
Let $\ll$ be an approximation relation on $\+B$.
Consider the relation $\lll$ on $\+C$ such that, for any $C,D\in\+C$,
$$
(*)\qquad
C\lll D \iff
\exists U,V\in\+B\ (C \supseteq U\ll V\supseteq D)\ .
$$
Let us check that $\lll$ satisfies conditions (1) to (4) of
Definition~\ref{def:approximation}.
Conditions (1) and (2) are straightforward.
We now look at Condition (3).
Suppose $x\in C\in\+C$. Let $U\in\+B$ be such that $x\in U\subseteq C$.
Applying condition (3) for $\ll$, there exists $V\in\+B$ such that
$x\in V$ and $U\ll V$.
Let $D\in\+C$ be such that $x\in D$ and $D\subseteq V$.
Then $C\lll D$ so that condition (3) holds for $\lll$.
Finally, we check Condition (4).
Suppose $C_i\lll C_{i+1}$ for all $i$.
Let $U_i,V_i\in\+B$ be such that
$C_i \supseteq U_i\ll V_i\supseteq C_{i+1}$.
In particular, $V_i \supseteq U_{i+1}\ll V_{i+1}$
hence $V_i\ll V_{i+1}$ by condition (2) for $\ll$.
Applying condition (4) for $\ll$, we see that
$\bigcap_{i\in\N}C_i=\bigcap_{i\in\N}V_i\neq\emptyset$
so that condition (4) holds for $\lll$.

In case $\ll$ is convergent then the $V_i$'s are a neighborhood basis
of some $x\in\bigcap_{i\in\N}V_i$.
Hence so are the $C_{i+1}$'s and $\lll$ is convergent.
Suppose $\ll$ is hereditary. 
To see that $\lll$ is also hereditary, observe that,
for any closed subset $F$ of $E$,
the relation $\lll_F$ is obtained from $\ll_F$ via condition $(*)$.
\end{proof}

\begin{proposition}\label{p:convergent implies hereditary}
Any convergent approximation relation is hereditary convergent.
\end{proposition}

\begin{proof}
Suppose $(U_i)_{i\in\N}$ is $\ll$-increasing and $C$ meets each $U_i$.
Since $\ll$ is a convergent approximation relation,
the $U_i$'s are a neighborhood basis of some
$x\in\bigcap_{i\in\N}U_i$.
Since the $U_i$'s meet $C$, $x$ is adherent to $C$, hence it is in $C$.
Also, the $U_i\cap C$'s are a neighborhood basis of $x$ in the subspace $C$.
\end{proof}

\subsection{Approximation Spaces versus quasi-Polish Spaces}
\label{ss:approx versus}

\begin{proposition}\label{p:approx}
Polish spaces and continuous domains are 
convergent approximation spaces.
\end{proposition}

\begin{proof}
{\it Case of Polish spaces.}
Let $\+B$ be the basis consisting of open balls centered in some fixed
countable dense set and having rational radius.
For $U,V\in\+B$, let $U\ll V$ if and only if
$\overline{V}\subseteq U$ and
$\textit{diam}(V) \leq \textit{diam}(U)/2\}$,
where  $\overline{V}$ is the topological closure of $V$
and \textit{diam} is the diameter.
All wanted conditions on $\ll$ are straightforward. 

{\it Case of continuous domains.}
Let $B$ be a basis in the sense of continuous domains.
The family
$\+B=\{\uup b\mid b\in B\textit{ and }\uup b\neq\emptyset\}$ is a topological basis.
Define the relation $\ll$ on $\+B$ as
$U\ll V$ if and only if $V=\uup c$ for some $c\in U$.
To check condition (1), observe that
if $U=\uup b$ then $c\in U$ yields $b\ll c$,
so that $U=\uup b\supseteq\uup c=V$.
As for condition (2),
if $W\supseteq U\ll V$ and $V=\uup c$ with $c\in U$
then $c\in W$ hence $W\ll V$.
As for condition (3), let $x\in U=\uup b$, i.e. $b\ll x$.
Using the interpolation property, let $c$ be such that
$b\ll c\ll x$ and set $W=\uup c$.
Then $x\in W$ and $c\in\uup b=U$, so that $U\ll W$.
Finally, for condition (4${}^+$), suppose $U_i\ll U_{i+1}$
for all $i\in\N$ and choose $b_{i+1}\in U_i$ such that
$U_{i+1}=\uup b_{i+1}$.
Since $b_{i+2}\in U_{i+1}=\uup b_{i+1}$,
we have $b_{i+1}\ll b_{i+2}$.
Let $x$ be the supremum of the $b_j$'s for $j\geq 1$.
Then $x\in\bigcap_{j\geq1}\uup b_j=\bigcap_{i\in\N}U_i$
which is therefore a non empty set.
Also, the $\uup b_j=U_j$'s, $j\geq1$,
are a basis of neighborhoods of $x$.
\end{proof}

Using de Brecht's Theorem~\ref{thm:deBrecht}
for second-countable spaces,
the above Proposition is partly subsumed by the next theorem. 

\begin{theorem}\label{thm:quasiPolish are approx}
Quasi-Polish spaces are convergent approximation spaces.
\end{theorem}

\begin{proof}
By de Brecht's result  stated in item 1 of Theorem~\ref{thm:deBrecht},
it suffices to show that any $\tbPi^0_2$ subspace $\+A$
of $\PN$ (with the Scott topology)
is a convergent approximation space.
In this proof we use Greek letters to denote finite sets.
For $\alpha\in\FIN$,
let $O_\alpha=\{X\mid \alpha\subseteq X\}$.
As a basis $\+B$ of the subspace $\+A$, we consider
those $B_\alpha=O_\alpha\cap\+A$, $\alpha\in\FIN$,
which are non empty.
Observe that  if $U\in\+B$, there may be infinitely many
$\alpha$'s such that $U=B_\alpha$.
The family $\+A$ is of the form
$\+A=\bigcap_{n\in\N}(U_n\cup F_n)$
where $U_n=\bigcup_{\alpha\in I_n}O_\alpha$ is open in $\PN$
and $F_n$ is closed in $\PN$.
Since $F_n$ is a countable intersection of closed sets of the
form $\PN\setminus O_\alpha$,
by merging this intersection with the global one,
we can reduce to the case where
$F_n=\PN\setminus O_{\alpha_n}$
with $\alpha_n\in\FIN$.
Then, 
$$X\in\+A \Leftrightarrow \forall n\ 
(\alpha_n\not\subseteq X \ \vee\
\exists \gamma\in I_n\ \gamma\subseteq X)
\Leftrightarrow \forall n\ 
(\alpha_n\subseteq X
\Rightarrow \exists \gamma\in I_n\ \gamma\subseteq X)\ .
$$
Let us call clause~$n$ the clause
$\alpha_n\subseteq X
\Rightarrow \exists \gamma\in I_n\ \gamma\subseteq X$.
Thus, a set $X\subseteq\N$ is in $\+A$ if and only if it
satisfies clause~$n$ for all $n$.
We introduce two notions:
\begin{itemize}

\item
Clause~$n$ {\em is a $U$-clause} if $U\subseteq O_{\alpha_n}$
(i.e. the premiss $\alpha_n\subseteq X$ of clause~$n$
is satisfied by all $X\in U$),

\item
Clause~$n$ {\em is $U$-solved}
if $U\subseteq O_\gamma$ for some $\gamma\in I_n$.
(i.e. the conclusion
$\exists \gamma\in I_n\ \gamma\subseteq X$ of clause~$n$
is satisfied by all $X\in U$ with the same witness $\gamma$).

\end{itemize}
Observe  that if $W\supseteq U$ then 
any $W$-clause is a $U$-clause and
any $W$-solved $W$-clause is a $U$-solved $U$-clause.
We denote by $n_U$ the least $n$ such that clause~$n$
is a $U$-unsolved $U$-clause or $+\infty$ if there is no
such clause.
We now define the relation $\ll$ on $\+B$~:
for $U,V\in\+B$,
$$
U\ll V \iff V\subseteq U\textit{ and }
\left\{\begin{array}{cl}
(i)&\textit{either $n_U=+\infty$}
\\
(ii)&\textit{or $n_U<+\infty$ and clause~$n_U$ is $V$-solved}
\\
(iii)&\textit{or $n_U<+\infty$ and, for some $m<n_U$,
clause~$m$ }
\\
&\textit{is not a $U$-clause
and is a $V$-solved $V$-clause.}
\end{array}
\right.
$$
We check conditions (1), (2), (3) and (4${}^+$) of Definition~\ref{def:approximation}.
Condition (1) is trivial.
Condition (2). Suppose $W\supseteq U\ll V$.
If $n_W=+\infty$ then $W\ll V$ holds by condition {\it(i)}.
So we shall suppose $n_W<+\infty$.
Suppose $n_W<n_U$.
Then clause $n_W$, being a $W$-clause hence a $U$-clause,
is $U$-solved (by definition of $n_U$),
hence it is also $V$-solved (since $U\supseteq V$),
so that $W\ll V$ holds by condition {\it(ii)}.
We now assume $n_U\leq n_W<+\infty$.
Since  $n_U$ is finite, $U\ll V$ cannot hold by condition {\it (i)}.
If $U\ll V$ holds by condition {\it(iii)} then the witnessing
clause $m$ is not a $U$-clause hence is not a $W$-clause,
so that $W\ll V$ holds by condition {\it(iii)}.
Suppose now that $U\ll V$ holds by condition {\it(ii)}.
If $n_W=n_U$ then $W\ll V$ also holds by condition {\it(ii)}.
Suppose $n_W>n_U$.
If clause $n_U$ were a $W$-clause then it would be $W$-solved
(by definition of $n_W$ and inequality $n_U<n_W$)
hence it would be a $U$-solved $U$-clause,
contradicting the definition of $n_U$.
Thus, clause $n_U$ is not a $W$-clause
and $W\ll V$ holds by condition {\it (iii)}.

Condition (3).
Let $X$ be in $U\in\+B$. 
If $n_U=+\infty$ then it suffices to set $V=U$.
Suppose now that $n_U<+\infty$.
Since clause~$n_U$ is a $U$-clause and $X\in U$,
we have $\alpha_{n_U}\subseteq X$
so that $X$ satisfies the premiss of clause $n_U$.
Then, $X$ being in $\+A$, satisfies all clauses,
in particular clause~$n_U$.
So, there exists some $\gamma\in I_{n_U}$
such that $\gamma\subseteq X$.
Now, $U\in\+B$ hence $U=\+A\cap O_\beta$ for
some $\beta\in\FIN$.
Set $V=U\cap O_\gamma$. Then
$V=\+A\cap O_\beta\cap O_\gamma
=\+A\cap O_{\beta\cup\gamma}\in\+B$.
Also, $X\in V$ since $X\in U$ and $\gamma\subseteq X$.
Finally, $V$ solves clause~$n_U$, hence~$U\ll V$
by condition (ii) of~$\ll$.

Condition (4).
Let $(U_i)_{i\in\N}$ be a sequence of
families in $\+B$ such that $U_i\ll U_{i+1}$ for all~$i$.
Let $\beta_i$ be any set in $\FIN$ such that
$U_i=\+A\cap O_{\beta_i}$.
Set $\delta_i=\bigcup_{j\leq i}\beta_j$
(so that the $\delta_i$'s are increasing with $i$).
Since $U_j\supseteq U_i$ for $j<i$,
we have $U_i=\+A\cap\bigcap_{j\leq i}O_{\beta_i}
=\+A\cap O_{\delta_i}$.
To finish the proof, we show that the set
$X=\bigcup_{i\in\N}\delta_i$ is in
the family $\bigcap_{i\in\N} U_i
=\+A\cap\bigcap_{i\in\N} O_{\delta_i}$
(which is therefore non empty).
Clearly, $X\in \bigcap_{i\in\N} O_{\delta_i}$.
To show that $X\in\+A$,
i.e. $X$ satisfies clause~$n$ for all $n$,
we argue by contradiction.
Suppose clause~$n$ is the first clause not satisfied by~$X$.
Then, $X$ satisfies clause $m$ for all $m<n$.
This means that if $\alpha_m\subseteq X$ then there is
$\gamma_{X,m}\in I_m$ such that $\gamma_{X,m}\subseteq X$.
Also, $X$ satisfies the premiss of clause~$n$ (but not its conclusion),
i.e. $\alpha_n\subseteq X$.
Since $\alpha_n$ and those $\alpha_m$'s, $\gamma_{X,m}$'s
included in $X$ (for $m<n$) are finite,
they are all included in $\delta_i$ for some~$i$.
Thus,

\begin{itemize}

\item
for each $m<n$, clause~$m$ is a $U_i$-clause
if and only if it is a $U_{i+1}$-clause
(if and only if $\alpha_m\subseteq X$),

\item
for each $m<n$, if clause~$m$ is a $U_i$-clause
then it is $U_i$-solved,

\item
clause~$n$ is a $U_i$-clause and is not $U_j$-solved for any $j\geq i$.
\end{itemize}
Then, clause~$n$ is the first $U_i$-unsolved $U_i$-clause
and no clause~$m$, $m<n$, can be a $U_{i+1}$-solved $U_{i+1}$-clause
without being a $U_i$-clause.
As a consequence, the assumed property $U_i\ll U_{i+1}$ 
necessarily comes from condition {\it (ii)} in the definition of $\ll$,
i.e. clause $n$ is  $U_{i+1}$-solved.
Since $X\in U_{i+1}$, clause $n$ is satisfied by~$X$. This is a contradiction.

Condition (4${}^+$).
To show that the $U_i$'s are a basis of neighborhoods of $X$
in $\+A$, it suffices to prove that, for all $\beta\in\FIN$
such that $X\in O_\beta$,
there exists $i$ such that
$\+A\cap O_\beta\subseteq\+A\cap O_{\delta_i}$.
Since $X=\bigcup_{i\in\N}\delta_i\in O_\beta$
we have $\beta\subseteq\delta_i$ for some $i$
hence $O_{\delta_i}\subseteq O_\beta$
and the wanted inclusion
$\+A\cap O_{\delta_i}\subseteq\+A\cap O_\beta$.
\end{proof}

\begin{example}\label{ex:Pinfini}
Applied to the quasi-Polish space $\INF$
and its basis $\{\+O_A\cap\INF \mid A\in\FIN\}$
(cf. Examples~\ref{ex:Scott examples}, and \ref{ex:domain}),
the previous proof gives an approximation relation which is not
the containment relation (cf. Example~\ref{ex:trivial approx}).
Clause~$n$ is $\exists j\geq n\ j\in X$ (i.e. $A_n$ is empty).
It is an $\+O_A\cap\INF$-clause and it is
$\+O_A\cap\INF$-solved if and only if $n\leq\max A$.
Then, for $A,B\in\FIN$, we have
$\+O_A\cap\INF\ll\+O_B\cap\INF$
if and only if
$A\subseteq B$ and $\max(A)<\max(B)$,
with the convention: $\max\emptyset=-1$.
\end{example}
The converse of Theorem~\ref{thm:quasiPolish are approx} is false.

\begin{proposition}\label{p:approx not hereditary}
There exists a Hausdorff approximation space with a countable basis
which is not a hereditary approximation space
hence neither convergent nor quasi-Polish
(by Proposition~\ref{p:convergent implies hereditary}
and Theorem~\ref{thm:quasiPolish are approx} ).
\end{proposition}

\begin{proof}
Consider the topology $\tau$ on the reals
generated by the usual open sets
and the set $\R\setminus\Q$ of irrational numbers.
A topological basis $\+B$ is
$\{]a,b[, ]a,b[\setminus\Q \mid a<b\}$.
Consider the relation $\ll$ on $\+B$ defined by the following clauses:
$$\begin{array}{rclcl}
]a,b[&\ll&]c,d[&\iff&a<c<d<b \textit{\quad(i.e. $[c,d]\subset]a,b[$)}\\
]a,b[&\ll&]c,d[\setminus\Q&\iff&a<c<d<b\\
]a,b[\setminus\Q&\ll&]c,d[\setminus\Q&\iff&
\theta(]a,b[\setminus\Q)\ll_{\textit{Baire}}\theta(]c,d[\setminus\Q)
\end{array}$$
where $\theta$ is the usual homeomorphism (given by continued fractions)
between $\R\setminus\Q$ and the Baire space $\Baire$
and $\ll_{\textit{Baire}}$ is any approximation relation on the family
of images by $\theta$ of bounded open real intervals
(cf. Propositions~\ref{p:approx} infra and
\ref{l:approximation and basis} supra).
It is easy to check that $\ll$ makes $(\R,\tau)$
an approximation space.
However, Theorem~\ref{thm:Pi2 Baire} infra ensures that
$(\R,\tau)$ is not a hereditary approximation space
since $\Q$ is closed in $(\R,\tau)$ and, as a subspace, has the usual
topology induced by $\R$ hence fails the Baire property.
\end{proof}

\begin{problem}
Suppose a space is ($T_0$ or $T_1$) hereditary approximation space 
and has a countable topological basis. Is it convergent (hence quasi-Polish)?
\end{problem}

\subsection{Characterization as Spaces of Choquet Games}

The next theorem asserts that approximation relations
on a topological basis of $X$ are essentially normalized strategies
for player Nonempty in the Choquet game $\Ch(X)$.

\begin{theorem}\label{thm:approx games}
If $X$ is a topological space with a well-orderable topological basis
(in particular, if $X$ has a countable basis)
then $X$ is an approximation (respectively convergent approximation) space
if and only if player Nonempty has a stationary (respectively and convergent)
winning strategy in the Choquet game $\Ch(X)$.
\end{theorem}

\begin{proof}
Suppose $\ll$ is an approximation relation
on some topological basis $\+B$ of $X$.
Fix some well-ordering of $\+B$.
Define a stationary strategy $\sigma$ for player Nonempty
in the Choquet game $\Ch(X)$ as follows.
If $U$ is a non empty set and $x\in U$ then define
$\sigma(x,U)=B$
where, letting $C$ be least in $\+B$ such that
$x\in C\subseteq U$,
applying condition (3) of Definition~\ref{def:approximation},
$B$ is least in $\+B$ such that $C\ll B$ and $x\in B$.
If 
$$(x_0,U_0),B_0,(x_1,U_1),B_1,\ldots$$ 
is a play where
Nonempty follows this strategy $\sigma$, there are open
sets $C_0,C_1,\ldots$ in $\+B$ such that
$$U_0\supseteq C_0\ll B_0
        \supseteq U_1\supseteq C_1\ll B_1\ldots$$
Applying condition (2),
we see that $B_0\ll B_1\ll\ldots$.
Condition (4) of Definition~\ref{def:approximation} ensures that
$\bigcap_{i\in\N}U_i=\bigcap_{i\in\N}B_i\neq\emptyset$.
Thus, $\sigma$ is a winning strategy for player Nonempty.
In case $\ll$ is convergent, condition (4${}^+$) ensures that
$\sigma$ is convergent.

Conversely, suppose that $\tau$ is a winning stationary strategy
for Nonempty in the Choquet game $\Ch(X)$.
In particular, $x\in\tau(x,U)\subseteq U$ for all open set $U$
and $x\in U$.
Fix some basis $\+B$ of $X$
and define a relation $\ll$ on this basis as follows:
for $B,C\in\+B$
$$
B\ll C
\iff
\exists D\in\+B\ \exists x\in D\ 
(B\supseteq D
\ \wedge\ x\in C\subseteq \tau(x,D))\ .
$$
Let us check the four conditions of Definition~\ref{def:approximation}.
Condition (1) is trivial since
$C\subseteq\tau(x,D)\subseteq D\subseteq B$.
Condition (2) is obvious from the definition of $\ll$.
As for condition (3), if $x\in B$ then $x\in\tau(x,B)$ so that,
for any $C\in\+B$ such that
$x\in C\subseteq \tau(x,B)$, we have $B\ll C$.
We now show condition (4).
Suppose that $B_n\ll B_{n+1}$ for all $n$.
Let $D_n\in\+B$, $x_n\in D_n$ be such that
$B_n\supseteq D_n$
and $x_n\in B_{n+1}\subseteq \tau(x_n,D_n)$.
Observe that the infinite sequence
$$(x_0,D_0),\tau(x_0,D_0),(x_1,D_1),\tau(x_1,D_1),
(x_2,D_2),\tau(x_2,D_2),\ldots$$
is a legal play of the Choquet game $\Ch(X)$ in which
Nonempty follows his winning strategy $\tau$.
Since $\tau(x_n,D_n)\supseteq B_{n+1}
                                   \supseteq D_{n+1}
                                   \supseteq \tau(x_{n+1},D_{n+1})$,
we have
$\bigcap_{n\in\N}B_n=\bigcap_{i\in\N}D_n\neq\emptyset$.
In case $\tau$ is a convergent stationary strategy,
it is clear that $\ll$ is a convergent approximation relation.
\end{proof}

\begin{remark}
The above proof fails if we consider the Banach-Mazur
game in place of the Choquet game:
we fail to obtain condition (3) for the relation $\ll$ associated to $\tau$.
\end{remark}

\subsection{Quasi-Polish and Second Countable
Convergent Approximation Spaces Coincide }
\label{ss:coincide}

\begin{theorem}\label{thm:coincide}
Quasi-Polish spaces coincide with $T_0$ convergent
approximation spaces with a countable basis.
\end{theorem}

\begin{proof}
By Theorem~\ref{thm:quasiPolish are approx},
quasi-Polish spaces are convergent approximation spaces.
If $X$ is a convergent approximation space then, 
using Theorem~\ref{thm:approx games},
Nonempty has a convergent winning strategy
in the Choquet game $\Ch(X)$.
By de Brecht's result (cf. item 3 of Theorem~\ref{thm:games}),
if $X$ is also $T_0$ and has a countable basis then it is quasi-Polish.
\end{proof}
Putting together Theorems~\ref{thm:coincide}, \ref{thm:approx games}
and \ref{thm:quasiPolish strategies} we can 
complement Corollary~\ref{coro:Markov stationary}.

\begin{corollary}\label{coro:Choquet games convergent}
Let $X$ be a $T_0$ space with a countable basis.
The following conditions  are equivalent:
\begin{enumerate}
\item Nonempty has a convergent wining strategy
in the Choquet game $\Ch(X)$,
\item Nonempty has a stationary convergent wining strategy
in the Choquet game $\Ch(X)$,
\item $X$ is a quasi-Polish space.
\end{enumerate}
\end{corollary}

\subsection{The $\tbPi^0_2$ Baire Property in Approximation Spaces}
\label{ss:Pi2Baire}

The Baire property holds in Polish spaces and
in compact  $T_2$  (i.e. Hausdorff) spaces, 
and it is also true in quasi-Polish spaces \cite{deBrecht2011}:
the intersection of countably many dense open sets is dense.
It trivially implies that the intersection of countably many
dense $\bGd$ sets is dense.
The next result is a formulation of the Baire property
for spaces where the $\tbPi^0_2$ and $\bGd$ classes do not coincide.
As a corollary, we see that the classical Baire property on
$\omega$-continuous domains and
also that on quasi-Polish
spaces (cf. Theorem~\ref{thm:deBrecht}) can be strengthened
from $\bGd$ to $\tbPi^0_2$.

\begin{theorem}\label{thm:Pi2 Baire}
If $X$ is an approximation space
with a well-orderable topological basis
(in particular, if $X$ has a countable basis)
then it satisfies the Baire property for $\tbPi^0_2(E)$ sets:
{\it the intersection of countably many dense $\tbPi^0_2(E)$ sets is dense}.
\end{theorem}

\begin{proof}
Using Lemma~\ref{l:approximation and basis},
consider a well-orderable basis $\+B$ and an approximation
relation $\ll$ on $\+B$.
We fix some well-order on $\+B$ and we will speak freely
of the least $U$ in $\+B$ satisfying a given property.
Since $\tbPi^0_2(E)$ sets are intersections of countably many
sets in $\co\tbD_2(E)$ (i.e. unions of an open and a closed sets),
and any superset of a dense set is dense,
it suffices to prove the Baire property for $\co\tbD_2$ sets.
Suppose $X=\bigcap_{n\in\N}U_n\cup F_n$
where, for each $n$, $U_n$ is open, $F_n$ is closed
and $U_n\cup F_n$ is dense.
We prove that $X$ is dense, i.e. $X$ meets $O$ for each $O\in\+B$.
Fix some $\pi:\N\to\N$ such that $\pi^{-1}(j)$ is infinite for each
$j\in\N$.
Fix some $O\in\+B$.
We inductively define a sequence of basic open sets $(O_n)_{n\in\N}$
such that $O\ll O_0$ and $O_n\ll O_{n+1}$ for all $n$.

(a) Let $O_0$ be the least $W\in\+B$ such that $O\ll W$
(there exists such a $W$~:
take any element $x$ in $O$ and apply condition (3) of Definition~\ref{def:approximation}).

(b) Suppose $n\geq1$, $O_{n-1}$ is defined and let $W\in\+B$
be least such that $O_{n-1}\ll W$ (as above there exists such a $W$).
If $U_{\pi(n)}$ contains some $W^*$ such that $W\ll W^*$
then set $O_n=W^*$. Otherwise, set $O_n=W$.
The following Claim finishes the proof of the Theorem.

{\it Claim.} The set $\bigcap_{n\in\N}O_n$ is non empty and
is included in $O \cap\left(\bigcap_{j\in\N}U_j\cup F_j\right)$.

{\it Proof of Claim.}
Since $\ll$ is transitive, $O_n\ll O_{n+1}$ holds for all $n$'s
and the non emptiness of $\bigcap_{n\in\N}O_n$ is ensured
by condition (4) of Definition~\ref{def:approximation}.
Let $x\in\bigcap_{n\in\N}O_n$.
By clause (a) we have $O\ll O_0$. Thus, $O_0\subseteq O$ and $x\in O$.
Let $j\in\N$.
Suppose $x\notin U_j$, we show that $x\in F_j$.
Since $x\notin U_j$ we have $O_n\not\subseteq U_j$ for all $n\in\N$.
Let $n$ be such that $\pi(n)=j$, clause (b) ensures that
$O_n=W$ (instead of $W^*$) and
\smallskip

(*) \qquad
$U_j$ contains no $V\in\+B$ such that $O_n\ll V$.
\smallskip
\\
Since $F_j$ is closed, to show that $x\in F_j$,
it suffices to prove that $x$ is adherent to $F_j$.
Consider $T\in\+B$ such that $x\in T$,
we have to prove that $T$ meets $F_j$.
Let $n$ be such that $\pi(n)=j$.
Since $U_j\cup F_j$ is dense and $T\cap O_n$
is non empty (it contains $x$), there exists
$x_n\in(U_j\cup F_j)\cap T\cap O_n$.
By way of contradiction, suppose $x_n$ is in $U_j$.
Then let $S\in\+B$ be such that
$x_n\in S\subseteq U_j\cap T\cap O_n$.
Condition (3) yields some $V\in\+B$ such that
$x_n\in V$ and $S\ll V$.
Using condition (2) and inclusion $S\subseteq O_n$
we have $O_n\ll V$.
Now, we also have $V\subseteq S\subseteq U_j$, 
contradicting (*).
Thus, $x_n\in F_j$; hence, $T$ meets $F_j$.
\end{proof}

\begin{remark}
The fact that every continuous domain $D$
satisfies the $\tbPi^0_2(D)$ Baire property
questions on a possible relation with
the Lawson topology.
Recall that the Lawson topology on a continuous domain is
the refinement of the Scott topology such that
the complement of the uppercone set $\up x$ is open, for each $x$ in $D$.
The Lawson topology is  $T_2$ and compact,
so it satisfies the usual Baire property with $\bGd$ sets.
Also, the Scott $\tbPi^0_2$ class is included in the Lawson class $\bGd$.
However, the $\tbPi^0_2$ Baire property for the Scott topology
differs from the usual $\bGd$ Baire property for the Lawson topology 
because the notion of dense set is not the same in the two topologies.
\end{remark}

\subsection{Hausdorff's Theorem and the $\tbPi^0_2$ Baire Property}
\label{ss:Hausdorff theorem}

Hausdorff-Kuratowski's theorem establishes that  for Polish spaces  the inclusion
in Proposition~\ref{p:Hausdorff Delta Borel} is an equality:
$\bigcup_{\alpha<\aleph_1}\tbD_\alpha(\tbSigma^0_\beta(E))
= \tbDelta^0_{\beta+1}(E)$.
One of the properties of Polish spaces used in the proof of this
result is the classical Baire property.
Matthew de Brecht (2011) has proved that Hausdorff-Kuratowski's theorem holds
for quasi-Polish spaces. 
Our next theorem gives conditions of  approximaton spaces to establish  the case $\beta=1$
i.e., Hausdorff's theorem
$\tbDelta^0_2(E)=\bigcup_{\alpha<\aleph_1}\tbD_\alpha(E)$.
The theorem pinpoints some topological properties
that suffice to prove this equality.
These properties are true in second countable hereditary
approximation spaces.
A similar result was obtained in
\cite{tang1981} for the Scott domain $\PN$.
Using a different proof  Selivanov (2005)  proved it  for
$\omega$-algebraic domains.

\begin{theorem}\label{thm:Hausdorff}
Let $E$ be a topological space satisfying the following properties.
\begin{enumerate}[(iii)]
\renewcommand{\theenumi}{(\roman{enumi})}

\item
There exists a countable basis of open sets.

\item
Every closed subspace $F$ of $E$ satisfies the
$\tbPi^0_2(E)$-Baire property: 
the intersection of countably many $\tbPi^0_2(E)$ subsets of $F$
which are dense in $F$ is also dense in $F$.
\end{enumerate}
Then,
$\bigcup_{\alpha<\aleph_1}\tbD_\alpha(E)= \tbDelta^0_2(E)$.
\end{theorem}
\begin{proof}
The $\subseteq$ inclusion is Proposition~\ref{p:Hausdorff Delta Borel}.
For the $\supseteq$ inclusion, we follow, mutatis mutandis,
Hausdorff's original proof with residues and adjoins, as exposed
in \cite{kechrisBook}, Theorem 22.27 pages~176-177.
For any subset $A\subseteq E$, we define by transfinite recursion a family
$(F_\alpha)_{\alpha<\aleph_1}$:
$$
F_0=E
\ ,\ F_{2\alpha+1}=\overline{A\cap F_{2\alpha}}
\ ,\ F_{2\alpha+2}=\overline{(E\setminus A)\cap F_{2\alpha+1}}
\ ,\ F_\lambda=\bigcap_{\alpha<\lambda} F_\alpha\quad\textit{if $\lambda$ is limit.}
$$
Applying property (i), let $(O_n)_{n\in\N}$ be a countable basis of open sets of $E$.
The $F_\alpha$'s are a decreasing sequence of closed sets.
If $F_{\xi+1}=F_\xi$ then $F_\alpha=F_\xi$ for all $\alpha\geq\xi$. 
If $F_{\alpha+1}$ is a strict subset of $F_\alpha$
then there is some $O_n$ which meets $F_\alpha$ but not $F_{\alpha+1}$.
Since there are countably many $O_n$'s,
this implies that there is some countable $\theta$ such that
$F_\theta=F_\alpha$ for all $\alpha\geq\theta$.
We shall consider the least even such $\theta$.

{\it Claim.} If $A$ is $\tbDelta^0_2(E)$ then $F_\theta=\emptyset$.

{\it Proof of Claim}. Suppose $F_\theta\neq\emptyset$.
Applying property (ii), the subspace $F_\theta$ is in $E$
hence it satisfies the Baire property.
Since $F_\theta=F_{\theta+2}$, we have
$F_\theta=\overline{A\cap F_\theta}=\overline{(E\setminus A)\cap F_\theta}$.
Thus, arguing in the subspace $F_\theta$,
the sets $A\cap F_\theta$ and $(E\setminus A)\cap F_\theta$ are
$\tbPi^0_2(F_\theta)$ dense subsets of $F_\theta$.
Since they are disjoint, this contradicts property (ii) in $F_\theta$.

Letting $\theta=2\zeta$ and
$B=\bigcup_{\alpha<\zeta}F_{2\alpha+1}\setminus F_{2\alpha+2}$,
we claim that $A=B$.
Indeed, suppose $x\in A$ and let $\eta$ be least such that
$x\notin F_\eta$. The inductive definition of the $F_\alpha$'s ensures
that $\eta=2\alpha+2$ for some $\alpha$.
Therefore, $x\in F_{2\alpha+1}\setminus F_{2\alpha+2}$  hence $x\in B$.
Similarly, if $x\notin A$ and $\eta$ is least such that
$x\notin F_\eta$ then $\eta=2\alpha+1$ for some $\alpha$.
Thus, $x\in F_{2\alpha}\setminus F_{2\alpha+1}$ hence $x\notin B$.
Observe that, for $\lambda$ limit, we have
$(E\setminus F_\lambda)\setminus
\bigcup_{\alpha<\lambda}(E\setminus F_\alpha)=\emptyset$.
Thus,
$A=B=\bigcup_{\alpha<\zeta}
(E\setminus F_{2\alpha+2})\setminus(E\setminus F_{2\alpha+1})
=D_{\theta+1}((E\setminus F_\alpha)_{\alpha\leq\theta})$
hence $A$ is in $\tbD_{\theta+1}(E)$.
\end{proof}

\begin{corollary}
Equality
$\displaystyle{\bigcup_{\alpha<\aleph_1}\tbD_\alpha(E)= \tbDelta^0_2(E)}$
holds in any hereditary approximation space having a countable basis.
\end{corollary}

\begin{problem}
What topological conditions ensure  higher levels of Hausdorff-Kuratowski theorem ?
\end{problem}

\section{The Hausdorff Hierarchy in Continuous Domains}
\label{s:Hausdorff domains}

Selivanov  made  a fine analysis of the Hausdorff difference hierarchy
in $\omega$-algebraic domains and he proved Hausdorff's theorem 
for these spaces (Selivanov 2005). He also showed the non existence 
of ambiguous sets in the Hausdorff hierarchy, provided there is a least 
element in the $\omega$-algebraic domain.
As stated in (Selivanov 2008, Theorem~3.4), the methods in his  paper of 2005
can be pushed from $\omega$-algebraic to $\omega$-continuous domains.

In this section we reconsider the question of the non existence of 
ambiguous sets in the Hausdorff hierarchy
for continuous domains (possibly not $\omega$-continuous).
We prove that the hypothesis of a  least element considered by Selivanov 
can be removed  for  successor levels of the Hausdorff hierarchy,
but not for limit levels.
Although  this improvement is a modest addition to Selivanov's result
(proofs being, mutatis mutandis, the same),
it requires to use the machinery of well-founded alternating trees developed
in \cite{selivanov2005}.
The extension from $\omega$-continuous domains to continuous domains
led us to consider possibly non countable alternating trees.
In fact, the countability of the domain basis is useful
only to ensure that the tree is countable, hence the rank of an alternating
tree is countable. But this hypothesis appears directly as an assumption
on the ordinal $\alpha$ in
Theorems~\ref{thm:Dalpha} and \ref{thm:no ambiguous degree}.

\subsection{Alternating Trees}
\label{ss:alternating trees}

First, we recall some simple notions about possibly uncountable trees.

\begin{definition}\label{def:tree}
\begin{enumerate}
\item
A {\it tree} is any non empty set of finite sequences closed under prefix.
The root of a tree is the empty sequence $\nil$.
If $\sigma$ is in a tree $T$ we let $T_\sigma$ be the tree
$\{\tau \mid \sigma\tau\in T\}$.
\item
A tree $T$ is {\it well-founded} if it has no infinite branch.
The ranks of the elements in a well-founded tree $T$ are defined inductively:
$\rank_T(\sigma)=\sup\{\rank_T(\sigma n)+1 \mid \sigma n\in T\}$
(convention for leaves of $T$ : $\sup \emptyset=0$).
The rank of $T$ is that of its root.
\item
An {\it alternating tree} is a map $f:T\to\Bool$ such that
$T$ is a tree and $f(\tau)\neq f(\sigma)$
whenever $\sigma$ is a son of $\tau$
(i.e. $\tau$ is a prefix of $\sigma$ and the length of $\sigma$ 
is the successor of that of $\tau$).
We say $f$ is $\varepsilon$-alternating if $f(\nil)=\varepsilon$.
\item
An {\it embedding} of an alternating tree $g:S\to\Bool$ into
an alternating tree $f:T\to\Bool$ is a monotone increasing
(with respect to the prefix ordering on finite sequences)
injective map $\theta:S\to T$ such that $g(\sigma)=f(\theta(\sigma))$.
We write $g\leq f$.
\end{enumerate}
\end{definition}
The next basic result is taken from \cite{selivanov2005}, Lemma 3.6 page 48.

\begin{lemma}\label{l:subtree}
Let $f:T\to\Bool$ be a well-founded alternating tree with rank $\alpha$.
For every $\beta<\alpha$ and every $\varepsilon\in\Bool$,
there exists an $\varepsilon$-alternating tree $g_\varepsilon\leq f$
with rank $\beta$.
\end{lemma}

\begin{proof}
We reproduce Selivanov's proof.
Argue by induction on $\alpha$.
If $\alpha$ is finite then take a branch of length $\alpha$ and remove
an appropriate tail and/or the root.
If $\alpha$ is limit then, for some $x$, $T_{(x)}$
(cf. Definition~\ref{def:tree}) has rank $\gamma$ such that $\beta<\gamma<\alpha$.
Use the induction hypothesis with $f_{(x)}:T_{(x)}\to\Bool$
such that $f_{(x)}(\sigma)=f(x\sigma)$.
Finally, suppose $\alpha=\lambda+m+1$ with $\lambda$ limit and $m\in\N$.
If $\beta<\lambda+m$ then let $x$ such that $T_{(x)}$ has rank $\lambda+m$
and use the induction hypothesis.
Suppose now $\beta=\lambda+m$.
Let $s$ with length $m+1$ such that $T_s$ has rank $\lambda$.
Let $\eta$ be the cofinality of $\lambda$ and
let $(\lambda_\xi)_{\xi<\eta}$ be strictly increasing with supremum $\lambda$
and $x_\xi$ be such that $T_{sx_\xi}$ has rank $\lambda_\xi$.
Let $g_{\xi,\delta}\leq f_{sx_{\xi+1}}$ for $\delta\in\Bool$
be $\delta$-alternating with rank $\lambda_\xi$.
Set $g_\varepsilon(0^p)=f(s\segment p)$ for $p\leq m$
and $g_\varepsilon(0^m\xi\sigma)=g_{\xi,\delta}(\sigma)$ 
where $\delta$ is chosen so as to get alternation from
$0^m$ to $0^m\xi$.
\end{proof}

\subsection{Alternating Trees and the Hausdorff Hierarchy}
\label{ss:alternating trees Hausdorff}

The following notion plays the role of intervals $[p,x]$ in \cite{selivanov2005}.

\begin{definition}\label{def:interval}
Let $b,x$ be elements of the domain $D$.
We let $\llbracket b,x]=\{y\mid b\ll y\leq x\}$.
\end{definition}

\begin{proposition}\label{p:interval}
Let $(D,\sqsubseteq)$ be a dcpo and $A\in\tbSigma^0_2(D)$.
If $(x_i)_{i\in I}$ is a directed system with supremum in $A$
then $x_i\in A$ for all $i$ large enough.
In particular, if $D$ is a continuous domain with basis $B\subseteq D$
and $x\in A$
then there exists $b\in B$ such that $b\ll x$ and $\llbracket b,x]\subseteq A$.
\end{proposition}

\begin{proof}
Let $A=\bigcup_{n\in\N}U_n\setminus V_n$ where $U_n, V_n$ are open.
Let $x=\sqcup_i x_i$. Since $x\in U_n$ and $U_n$ is Scott open,
the $x_i$'s are in $U_n$ for $i$ large enough.
Since $x\notin V_n$ and $V_n$ is an upset, no $x_i$ is in $V_n$.
Thus, the $x_i$'s are in $U_n\setminus V_n$ hence in $A$ for $i$ large enough.
\end{proof}

To keep the presentation self-contained, we prove the extensions to continuous domains
of Theorems 3.10 and 3.14 in \cite{selivanov2005}.
Proofs are almost the same:
compact elements are replaced by elements of some fixed basis $B$
and the countability of the basis
(which ensures the countability of the rank of alternating $B$-trees, cf. Definition~\ref{def:A alt})
is replaced by the assumed countability of the considered ordinal $\alpha$.

\begin{definition}\label{def:A alt}
Let $(D,\sqsubseteq)$ be a dcpo and $B\subseteq D$.
If $A$ is any subset of $D$ we let $\chi_A:A\to\Bool$
be the characteristic function of $A$ (which takes value $1$ on $A$).
\begin{enumerate}
\item
A $B$-tree is a map $f:T\to B$ where $T$ is a tree.
\item
A $B$-tree $f:T\to B$ is $(A,\varepsilon)$-alternating if the tree
$\chi_A\circ f:T\to\Bool$ is $\varepsilon$-alternating
in the sense of Definition~\ref{def:tree}, i.e.
$f(\sigma n)\in A\Leftrightarrow f(\sigma)\notin A$
for all $n\in\N$ and $\sigma,\sigma n\in T$.
\item
Let $\preceq$ be $\sqsubseteq$ or $\ll$.
A $B$-tree $f$ is $\preceq$-increasing
if $f(\sigma)\preceq f(\sigma n)$ for all $\sigma,\sigma n\in T$.
\end{enumerate}
\end{definition}

\begin{proposition}\label{p:no alt omega chain}
Let $(D,\sqsubseteq)$ be a dcpo and $B\subseteq D$.
If $A$ is $\tbDelta^0_2(D)$ then every $A$-alternating 
$\sqsubseteq$-increasing $B$-tree $f:T\to B$ is well-founded.
\end{proposition}

\begin{proof} 
Let $(n_i)_{i\in\N}$ be an infinite branch of $f$.
The $f(n_0\ldots n_i)$'s are $\sqsubseteq$-increasing.
Let $x$ be their supremum.
Suppose $x\in A$. Proposition~\ref{p:interval} establishes that
$f(n_0\ldots n_i)$ is in $A$ for all $i$ large enough,
contradicting alternation of $f$.
Idem if $x\notin A$. 
\end{proof}
The next result is a slight variant of
(Selivanov 2005, Proposition 3.8 and Theorem 3.10).

\begin{theorem}\label{thm:Dalpha}
Let $(D,\sqsubseteq)$ be a dcpo, $B\subseteq D$,
$A\in\tbDelta^0_2(D)$ and $0<\alpha<\aleph_1$.
\begin{enumerate}
\item
If $A\in\tbD_\alpha(D)$ then there is no $\leq$-increasing
$(A,1)$-alternating $B$-tree with rank $\alpha$.

\item
If $D$ is a continuous domain with basis $B$,
the following conditions are equivalent.

\begin{enumerate}

\item[$i.$]\
$A\in\tbD_\alpha(D)$,
\item[$ii.$]\
there is no $\sqsubseteq$-increasing
$(A,1)$-alternating $B$-tree with rank $\alpha$.
\item[$iii.$]\
there is no $\ll$-increasing
$(A,1)$-alternating $B$-tree with rank $\alpha$.
\end{enumerate}

\end{enumerate}
\end{theorem}

\begin{proof}
1. We argue by induction on $\alpha$.
Let $A=D_\alpha((A_\beta)_{\beta<\alpha})$ where the $A_\beta$'s
are a monotone increasing $\alpha$-sequence of open sets.
By way of contradiction, suppose there exists an increasing $(A,1)$-alternating
$B$-tree $f:T\to B$ with rank $\alpha$.
Since $f(\nil)\in A$ we have
$f(\nil)\in A_\beta\setminus\bigcup_{\gamma<\beta}A_\gamma$
for some $\beta<\alpha$, $\beta\not\sim\alpha$.
Since $f(\nil)\leq f(\sigma)$ for all $\sigma\in T$ and $A_\beta$ is open
(hence, an upset),
we see that the range of $f$ is included in $A_\beta$.
If $(n)\in T$ then $f((n))\notin A$ hence
$f((n))\in A_\beta\setminus A\subseteq\bigcup_{\gamma<\beta}A_\gamma$.
Since $\bigcup_{\gamma<\beta}A_\gamma$ is open (hence, an upset)
we see that $f(\sigma)\in\bigcup_{\gamma<\beta}A_\gamma$ for all
$\sigma\in T$, $\sigma\neq\nil$.
Let $A^-=D_\alpha((A_\gamma)_{\gamma<\beta})$.
Then $f$ is an $(A^-,0)$-alternating $B$-tree.
Now, $f$ has rank $\alpha$ and Lemma~\ref{l:subtree} implies that there exists
$g\leq f$ which is an $(A^-,1)$-alternating $B$-tree with rank $\beta$.
Since $A^-$ is in $\tbD_\beta$, the inductive hypothesis is contradicted.

2. Since $i\Rightarrow ii$ is item 1 and $ii\Rightarrow iii$ is trivial,
it remains to prove $iii\Rightarrow i$.
For each $b\in B$, let $S_b$ be the family of
finite sequences $(b,b_1,\ldots,b_k)$ of elements of $B$
satisfying the following conditions:

\begin{itemize}
\item 
$b\ll b_1\ll\ldots\ll b_k$,

\item 
$b_i\in A\Leftrightarrow b_{i+1}\notin A$ for all $0\leq i<k$
(with $b_0=b$).
\end{itemize}
Fix some bijection $\theta$ between $B$ and an initial segment of $\N$.
Applying $\theta$, transform $S_b$ into an $A$-alternating $B$-tree
$f_b:T_b\to B$ such that
$$T_b=\{\nil\}\cup\{(\theta(b_1),\ldots,\theta(b_k)) \mid
                (b,b_1,\ldots,b_k) \in S_b\}$$
and $f_b(\nil)=b$, $f_b((\theta(b_1),\ldots,\theta(b_k)))=b_k$.
Since $A$ is assumed to be $\tbDelta^0_2(D)$,
Proposition~\ref{p:no alt omega chain} ensures that $f_b$ is well-founded.
Suppose there is no $(A,1)$-alternating $B$-tree with rank $\alpha$.
Then $f_b$ has rank $\leq\alpha$ if $b\notin A$
and rank $<\alpha$ if $b\in A$.
For $\beta<\alpha$ define the open sets
$$
A_\beta=
\bigcup\{\uup f_b(\sigma) \mid b\in B,\ \rank_{T_b}(\sigma)\leq\beta,\ 
f_b(\sigma)\in A\Leftrightarrow \beta\not\sim\alpha\}.
$$
To conclude we prove that $A=\tbD_\alpha((A_\beta)_{\beta<\alpha})$.
First, we show that $A\subseteq\bigcup_{\beta<\alpha}A_\beta$.
Suppose $x\in A$. Applying Proposition~\ref{p:interval}, we get
$\llbracket b,x]\subseteq A$ for some $b\ll x$, $b\in B$.
Since $b\in A$, we have $\rank(T_b)<\alpha$.
A fortiori, $\rank(T_b)\leq\beta$ with $\beta<\alpha$
and $\beta\not\sim\alpha$.
Since $f_b(\nil)=b$ we have $\uup b\subseteq A_\beta$ hence $x\in A_\beta$.
For $\beta<\alpha$, let
$A'_\beta=A_\beta\setminus\bigcup_{\gamma<\beta}A_\gamma$.
Since $A\subseteq\bigcup_{\beta<\alpha}A_\beta$,
to show $A=\tbD_\alpha((A_\beta)_{\beta<\alpha})$
it suffices to prove that
$\beta\not\sim\alpha \Rightarrow A'_\beta\subseteq A$ and
$\beta\sim\alpha \Rightarrow A'_\beta\subseteq D\setminus A$.

{\it Case $\beta\not\sim\alpha$.}
By way of contradiction, suppose $A'_\beta\not\subseteq A$ and
let $x\in A'_\beta\setminus A$. 
By Proposition~\ref{p:interval}, we have
$\llbracket c,x]\subseteq D\setminus A$ for some $c\ll x$, $c\in B$.
Now, since $x\in A_\beta$, there exist $b\in B$ and $\sigma\in T_b$
such that $x\in\uup f_b(\sigma)$, $\rank_{T_b}(\sigma)\leq\beta$,
$f_b(\sigma)\in A$.
Since $c,f_b(\sigma)\ll x$, the interpolation property gives an
$e\in B$ such that $c,f_b(\sigma)\ll e \ll x$.
Since $e\in\llbracket c,x]$ we have $e\notin A$.
Let $\sigma=(\theta(b_1),\ldots,\theta(b_k))$ where
$(b,b_1,\ldots,b_k)\in S_b$.
Since $b_k=f_b(\sigma)\in A$, $f_b(\sigma)\ll e$ and $e\notin A$,
the sequence $(b,b_1,\ldots,b_k,e)\in S_b$.
Hence $\sigma\theta(e)\in T_b$.
Now, $\rank(\sigma\theta(e))<\rank(\sigma)\leq\beta$.
Since $\beta\not\sim\alpha$, there is some $\gamma<\beta$
such that $\gamma\sim\alpha$ and
$\rank(\sigma\theta(e))\leq\gamma$.
Summing up, we have $\uup e\subseteq A_\gamma$.
Since $e\ll x$ we get $x\in A_\gamma$
which contradicts $x\in A'_\beta$.

{\it Case $\beta\sim\alpha$.}
The proof that $A'_\beta\subseteq D\setminus A$ is similar.
\end{proof}

\subsection{Ambiguous Sets in the Hausdorff Hierarchy}
\label{ss:ambiguous}

We now come to the question of whether there are ambiguous sets in the Hausdorff hierarchy.
Item 1 of the next Theorem was obtained by Selivanov (2005)
for $\omega$-algebraic domains.

\begin{theorem}\label{thm:no ambiguous degree}
Let $D$ be a continuous domain and $0\leq\alpha<\aleph_1$.
\begin{enumerate}
\item
If $D$ has a least element $\bot$ then 
$\tbD_\alpha(D)\cap\co\tbD_\alpha(D)
= \bigcup_{\beta<\alpha}\tbD_\beta(D)\cup\co\tbD_\beta(D)$.

\item
$\tbD_{\alpha+1}(D)\cap\co\tbD_{\alpha+1}(D)
= \bigcup_{\beta\leq\alpha}\tbD_\beta(D)\cup\co\tbD_\beta(D)$
for all $1\leq\alpha<\aleph_1$.

\item
In general, equality
$\tbD_\alpha(D)\cap\co\tbD_\alpha(D)
= \bigcup_{\beta<\alpha}\tbD_\beta(D)\cup\co\tbD_\beta(D)$
fails for $\alpha=1$ and for $\alpha$ limit.
\end{enumerate}
\end{theorem}

\begin{proof}
1. Inclusion right to left comes from Proposition~\ref{p:Hausdorff}.
Inclusion left to right is proved by induction on $\alpha$.
Case $\alpha=0$ is trivial since both members are empty.

Suppose $\alpha=\beta+1$ and
$A\in\tbD_\alpha(D)\cap\co\tbD_\alpha(D)$.
Towards a contradiction, suppose
$A\notin\tbD_\beta(D)\cup\co\tbD_\beta(D)$,
i.e. neither $A$ nor $D\setminus A$ is in $\notin\tbD_\beta(D)$.
Theorem~\ref{thm:Dalpha} proves the existence,
for $\varepsilon=0,1$, of an increasing
$(A,\varepsilon)$-alternating $B$-tree
$f_\varepsilon:T_\varepsilon\to B$ with rank $\beta$.
Observe that every domain basis contains $\bot$, hence $\bot\in B$.
Let $\tau=0$ if $\bot\in A$ and $\tau=1$ otherwise,
so that $f_\tau(\nil)\neq\bot$.
Let $S=\{0\sigma\mid\sigma\in T_\tau\}$
and $g:S\to B$ be such that $g(\nil)=\bot$
and $g(0\sigma)=f_\varepsilon(\sigma)$.
Then $g$ is an increasing $(A,1-\tau)$-alternating
$B$-tree with rank $\beta+1=\alpha$.
Applying again Theorem~\ref{thm:Dalpha},
if $\tau=0$ this contradicts $A\in\tbD_\alpha(D)$,
if $\tau=1$ this contradicts $A\in\co\tbD_\alpha(D)$. 

Suppose now $\alpha$ is a limit ordinal and $(\alpha_n)_{n\in\N}$
is increasing with $\alpha$ as supremum.
Towards a  contradiction, suppose
$A\notin\bigcup_{\beta<\alpha}\tbD_\beta(D)\cup\co\tbD_\beta(D)$.
Then, for every $n$ and $\varepsilon=0,1$, there is
a monotone increasing
$(A,\varepsilon)$-alternating $B$-tree
$f_{n,\varepsilon}:T_{n,\varepsilon}\to B$ with rank $\alpha_n$.
Let $\tau\in\{0,1\}$ be as above.
Set $S=\{n\sigma\mid\sigma\in T_{n,\varepsilon}\}$
and $g:S\to B$ be such that $g(\nil)=\bot$
and $g(n\sigma)=f_{n,\varepsilon}(\sigma)$.
Then $g$ is a monotone increasing $(A,\tau)$-alternating
$B$-tree with rank $\alpha$.
As above, this gives a contradiction.

2. Let $A=D_{\alpha+1}((A_\beta)_{\beta\leq\alpha})$
and $D\setminus A=D_{\alpha+1}((E_\beta)_{\beta\leq\alpha})$
where the $A_\beta, E_\beta$'s are $\alpha+1$ increasing sequences
of open sets.
Set $D^+=D\cup\{\bot_\alpha\}$ and extend the order of $D$
by setting
$\bot_\alpha<x$ for all $x\in A_\alpha$.
Set $A^+=A\cup\{\bot_\alpha\}$.
Then, in $D^+$, we have
$A^+=D_{\alpha+1}((A^*_\beta)_{\beta\leq\alpha})$
where $A^*_\alpha=A_\alpha\cup\{\bot_\alpha\}$
and $A^*_\beta=A_\beta$ for $\beta<\alpha$.
Also, $D^+\setminus A^+=D\setminus A$.
Observe that the $A^*_\beta$'s and $E_\beta$'s are open in $D^+$.
Thus, $A^+$ is ambiguous at level $\alpha+1$ in $D^+$.
Though $\bot_\alpha$ is not a least element in $D^+$,
it is smaller than $A_\alpha$ hence smaller than all
elements labeling an $(A,1)$-alternating $B$-tree with rank $\alpha$.
Arguing as in item 1, we see that $A^+$ must have level  at most $\alpha$
in $D^+$. 
If $D^+\setminus A^+=D\setminus A$
has level $\alpha$ in $D^+$ then $D\setminus A$ is obtained via
open sets not containing $\bot_\alpha$ hence $D\setminus A$
has level $\alpha$ in $D$.
Suppose now that $A^+$ has level $\alpha$ in $D^+$.
Then $A^+=D_{\alpha}((C_\delta)_{\delta<\alpha})$
where the $C_\delta$'s are open in $D^+$.
Since $\bot_\alpha\in A^+$ there is some $\beta<\alpha$
with parity different from $\alpha$ such that
$\bot_\alpha\in C_\beta$ hence
$C_\beta\supset\bigcup_{\delta<\alpha}A_\delta$.
This yields 
$A=D_{\alpha}((C^-_\gamma)_{\gamma<\alpha})$
where $C^-_\gamma=C_\gamma$ for $\gamma<\beta$
and $C^-_\gamma=C_\gamma\cap D$ for $\gamma\geq\beta$.
So, $A$ has level $\beta$ in $D$.
Finally, observe that the argument breaks down if $\alpha=0$.

3. In $\mix$, the set $0\mix$ is in
$\tbD_1(D)\cap\co\tbD_1(D)$
but not in $\tbD_0(D)\cap\co\tbD_0(D)$.
Let $\alpha$ be a countable limit ordinal.
Let $(a_\beta)_{\beta<\alpha}$ be a strictly decreasing
$\alpha$-sequence of reals with no lower bound
and which is continuous:
$a_\lambda=\inf_{\delta<\lambda}a_\delta$
for all limit $\lambda<\alpha$.
In $\overrightarrow{\R}$
(cf. Example~\ref{ex:Scott examples}),
consider the set
$A=\tbD_\alpha((A_\beta)_{\beta<\alpha})$
where $A_\beta=]a_\beta,+\infty]$.
Then $\overrightarrow{\R}\setminus A=
\tbD_\alpha((A^*_\beta)_{\beta<\alpha})$
where $A^*_0=\emptyset$ and
$A^*_{\beta+1}=A_\beta$ and $A^*_\lambda=A_\lambda$
for $\lambda$ limit.
Thus, $A$ is ambiguous at level $\alpha$.
It is easy to check that $A$ is not in a lesser level.
\end{proof}

\section{Effective Borel and Hausdorff Hierarchies}
\label{s:effective hierarchies}

\subsection{Effective Topological Spaces}\label{ss:effective space}

The first step to deal with the effective Borel hierarchy
is the definition of effective topological space.
We follow Weihrauch's book (2000),
Definition 3.2.1, page 63.

\begin{definition}
An {\it effective topological space} is a pair $(E,(O_n)_{n\in\N})$ where
$E$ is a topological space admitting a countable basis
and $(O_n)_{n\in\N}$ is an enumeration (not necessarily injective)
of some topological basis of $E$.
\end{definition}

\begin{remark}
For the notion of {\it computable topological space}, one also requires that
the equivalence relation $\{(m,n)\mid O_m=O_n\}$ be computably enumerable.
For instance, this is the case if $n\mapsto O_n$ is injective,
which is usually true. We shall not need this notion.
\end{remark}

\begin{definition}\label{def:approximation eff}
An effective  approximation space is a triple
$(E,(O_n)_{n\in\N},\ll)$ such that $(O_n)_{n\in\N}$
enumerates a topological basis $\+B$,
the relation $\{(i,j)\mid O_i\ll O_j\}$ is computably enumerable
and $\ll$ is an approximation relation on $\+B$
(cf. Definition~\ref{def:approximation}).
\end{definition}

\subsection{Effective $\omega$-Continuous Domains}
\label{ss:effective domains}

\begin{definition}\label{def:computableScottdomain}
An $\omega$-continuous domain is {\it effective}
if it admits a basis $B=\{b_n \mid n\in\N\}$ such that
$\{(i,j) \mid b_j\ll b_i\}$ is computably enumerable.
\end{definition}

\begin{example}\label{def:examplescomputableScottdomain}
$(\INF,\subseteq)$ is not a continuous domain
(cf. Example~\ref{ex:Pinfini}) but it is an effective approximation space.
Other spaces in Example~\ref{ex:Scott examples} are effective $\omega$-continuous domains.
\end{example}

\begin{proposition}\label{p:eff topo approx}
Every effective $\omega$-continuous domain is an effective
approximation space (hence an effective topological space).
\end{proposition}

\begin{proof}
Immediate from the proof of Proposition~\ref{p:approx}.
\end{proof}

\subsection{Borel Codes}\label{ss:Borel codes}

There are several ways to code Borel sets by elements in the Baire space
$\Baire$, cf. \cite{moschovakisbook} \S3H and 7B,
\cite{marker2002} \S7.
We choose a coding adapted to the context of effective topological spaces.

\begin{definition}
Let $(E,(O_n)_{n\in\N})$ be an effective topological space.
We code Borel sets by well-founded trees
(cf. Definition~\ref{def:tree}).
To any $\sigma$ in a well-founded tree $T$
we attach a Borel subset $\means{\sigma}$ of $E$
by induction on the rank.
\begin{enumerate}[(iii)]
\renewcommand{\theenumi}{(\roman{enumi})}

\item
If $\rank_T(\nil)=0$ (i.e. the tree $T$ is reduced to its root)
then $\means{\nil}_T=\emptyset$,

\item
If $\rank_T(\sigma)=0$ and $\sigma\neq\nil$ and $\sigma$ has last element $n$
then $\means{\sigma}_T=O_n$,

\item
If $\rank_T(\sigma)=1$ then
$\means{\sigma}_T
=\bigcup_{n\in\N:\sigma n\in T} O_n$.

\item
If $\rank_T(\sigma)\geq2$ then
$\means{\sigma}_T
=\bigcup_{n\in\N:\sigma\conc 2n,\sigma\conc2n+1\in T}
\means{\sigma\conc2n}_T\setminus\means{\sigma\conc2n+1}_T$.

\end{enumerate}
The $\Sigma$-Borel and $\Pi$-Borel sets coded by $T$ are
$\means{T}_\Sigma=\means{\nil}_T$
and $\means{T}_\Pi=E\setminus\means{\nil}_T$.
\end{definition}
Observing that the above inductive definition of $\means{\sigma}_T$
follows exactly that of the Borel hierarchy, a straightforward induction
on the ordinal $\alpha$ shows the following result.

\begin{proposition}
Let $(E,(O_n)_{n\in\N})$ be an effective topological space.
A subset of $E$ is in $\tbSigma^0_\alpha(E)$ (respectively $\tbPi^0_\alpha(E)$)
if and only if it is of the form
$\means{T}_\Sigma$ (respectively $\means{T}_\Pi$)
for some well-founded tree with rank at most $\alpha$.
\end{proposition}

\subsection{The Effective Borel Hierarchy}\label{ss:effective Borel}

Borel codes lead to a definition of the  effective Borel hierarchy.
First, we recall some classical results about computable ordinals
which imply that there is a huge latitude to represent them:
from computability with very low resource complexity
up to hyperarithmeticity.

\begin{proposition}\label{p:omega1CK}
There exists an ordinal $\CK$
(the Church-Kleene ordinal)
such that, for every countable ordinal $\alpha$,
the following properties are equivalent.
\begin{enumerate}[(iii)]
\renewcommand{\theenumi}{(\roman{enumi})}

\item
$\alpha<\CK$,

\item
$\alpha$ is the rank of some computable well-founded tree,

\item
$\alpha$ is the rank of some hyperarithmetical (i.e. $\Delta^1_1$)
well-founded tree~\cite{spector1955},

\item
$\alpha$ is the order type of some computable linear order on $\N$,
i.e. $\alpha$ is computable,

\item
$\alpha$ is the order type of some hyperarithmetical linear order on $\N$,
i.e. $\alpha$ is $\Delta^1_1$,

\item
$\alpha$ is the order type of some linear order on $\N$
which is computable in real time and logarithmic space~\cite{dehornoy1986,grigorieff1990}.
\end{enumerate}
Moreover, in (iv-vi), one can suppose that, for each $n\in\N$,
the set of elements with rank exactly $n$
and that with ranks in
$\{\omega\alpha+n \mid \alpha<\CK\}$
are computable.  
\end{proposition}

\begin{remark}\label{rk:omega1CK}
The last assertion in Proposition~\ref{p:omega1CK} is a simple trick
in~\cite{ershov1968}.
If $(\N,R)$ has type $\alpha$ then the lexicographic product of
$(\N,R)$ and $(\N,<)$ has type $\omega\alpha$ and the (computable) set
$\N\times\{n\}$ consists of all elements of rank $\equiv n\bmod\omega$.
\end{remark}
The definition of the effective version of Borel hierarchy for countably based spaces,
and the basic properties appeared previously in \cite{selivanov2008}.

\begin{definition}[Effective Borel hierarchy]
\label{def:effBorelfinite}
Let $(E,(O_n)_{n\in\N})$ be an effective topological space
and $\alpha$ an ordinal such that $1\leq\alpha<\CK$.
\begin{enumerate}
\item
The {\it effective Borel classes}
$\Sigma_\alpha^0(E)$, $\Pi_\alpha^0(E)$, $\Delta_\alpha^0(E)$,
are defined as follows:
\begin{itemize}

\item A set is in the class $\Sigma_\alpha^0(E)$ (respectively $\Pi_\alpha^0(E)$)
if and only if it is of the form
$\means{T}_\Sigma$ (respectively $\means{T}_\Pi$)
for some well-founded tree $T$ with rank at most $\alpha$
such that both $T$ and the rank order relation on $T$
(i.e. $\{(s,t)\in T\times T\mid \rank_T(s)\leq\rank_T(t)\}$) are computable.

\item The class $\Delta_\alpha^0(E)=\Sigma_\alpha^0(E)\cap\Pi_\alpha^0(E)$.
\end{itemize}

\item
A sequence $(X_n)_{n\in\N}$ of subsets of $E$
is uniformly $\Sigma_\alpha^0(E)$ if there exists a computable sequence
of well-founded trees $(T_n)_{n\in\N}$ with ranks at most $\alpha$
such that $X_n=\means{T_n}_\Sigma$.
Idem with uniformly $\Pi_\alpha^0(E)$.

\item
$G_\delta$ and $F_\sigma$ are the classes of intersections of
uniformly $\Sigma_1^0(E)$ sequences
(respectively unions of uniformly $\Pi_1^0(E)$ sequences).
\end{enumerate}
\end{definition}

\begin{remark}\label{rk:Borel effective}
\begin{enumerate}
\item  If $\preceq$ is an order on $\N$ isomorphic to the ordinal $\alpha$
then the set of sequences $(n_1,\ldots,n_k)$ such that
$n_1\succeq n_2 \ldots \succeq n_k$ is a well-founded tree with rank $\alpha$
and computable rank order relation (i.e. it satisfies the above condition (a)).

\item
 The requirement  that the rank order relation on $T$
is computable (in condition 1)
allows to get the usual definition of $\Sigma^0_\alpha(E)$
for finite $\alpha$'s.
For instance, a subset $X$ of $E$ is $\Sigma^0_2(E)$
(respectively $\Sigma^0_3(E)$) if and only if
there exists computably enumerable sets $A,B\subseteq\N^2$
(respectively $A,B\subseteq\N^3$)
such that 
$X=\bigcup_{i\in\N}
\left(\bigcup_{j: (i,j)\in A}O_j\right)
\setminus
\left(\bigcup_{j:(i,j)\in B}O_j\right)$
\\
(respectively,
$X=\bigcup_{i\in\N}\bigcap_{j\in\N}
\left(E\setminus\bigcup_{k:(i,j,k)\in A}O_k\right)
\cup
\left(\bigcup_{k: (i,j,k)\in B}O_k\right)$).
\end{enumerate}
\end{remark}

\begin{proposition}\label{p:Borel effective}
All assertions in Proposition~\ref{p:Borel}
hold for the effective Borel
classes with the proviso that
$1\leq\alpha<\CK$
and countable unions (respectively intersections)
are relative to uniformly $\Sigma_\alpha^0(E)$ (respectively $\Pi_\alpha^0(E)$)
sequences of sets.
\end{proposition}

\subsection{Hausdorff Codes and the Effective Hausdorff Hierarchy}\label{ss:effective Hausdorff}

The definition of the effective difference hierarchy for countably based spaces
and their basic properties appeared previously in \cite{selivanov2008}.

\begin{definition}\label{def:eff Hausdorff}
Let $(E,(O_n)_{n\in\N})$ be an effective topological space.
\begin{enumerate}
\item A Hausdorff $\alpha$-code for a set $X$ in $\tbD_\alpha(\tbSigma^0_\beta(E))$
is a triple $(\preceq,P,(T_n)_{n\in\N})$ such that

\begin{itemize}
\item
$\preceq$ is a well-order of type $\alpha$ on $\N$
or on a finite initial segment of $\N$,
\item
$P=\{n\in\dom(\preceq) \mid \varphi(n)\sim\alpha\}$
where $\varphi$ is the unique isomorphism from $(\dom(\preceq),\preceq)$
onto the ordinal $\alpha$,
\item
$(T_n)_{n\in\N}$ is a family of well-founded trees
with ranks at most $\beta$,
\item
$X=\bigcup_{p\in P} \means{T_p}_\Sigma 
\setminus \cup_{\varphi(q)<\varphi(p)} \means{T_q}_\Sigma$.
\end{itemize}

\item The effective Hausdorff classes $\tlD_\alpha(\Sigma^0_\beta(E))$
are defined as follows:
for $1\leq\alpha<\CK$,
a set $X$ is in $\tlD_\alpha(\Sigma^0_\beta(E))$
if and only if it admits an $\alpha$-code $(\preceq,P,(T_n)_{n\in\N})$
such that $\preceq$ and $P$ are computable
and the $T_n$'s and the rank relations on the $T_n$'s are uniformly computable,
(i.e. $\{(n,s) \mid s\in T_n\}$ and 
$\{(n,s,t)\mid s,t\in T_n, \rank_{T_n}(s) \leq \rank_{T_n}(t)\}$
are computable).
\end{enumerate}
The effective class $\tlD_\alpha(\Sigma^0_1(E))$ is also denoted by
$\tlD_\alpha(E)$.
\end{definition}

\begin{proposition}\label{p:Hausdorff Delta Borel effective}
Propositions~\ref{p:Hausdorff},
\ref{p:Hausdorff Delta Borel}
and \ref{p:inverseimagesHausdorff}
hold with the effective Hausdorff classes and effectively continuous maps.
\end{proposition}

\subsection{Does Hausdorff's Theorem Fully Effectivize?}
\label{ss:Hausdorff eff}

\begin{problem}
Equality
$\displaystyle{\bigcup_{\alpha<\omega_1^{CK}}\tlD_\alpha(E) = \Delta^0_2(E)}$
holds in computable Polish spaces,
(cf. Selivanov 2003, pages 76-79 for the Baire space).
Is this also true for more general spaces including
effective $\omega$-continuous domains
endowed with the Scott topology?
\end{problem}

Contrary to what was the case with the other results,
the proof of Theorem~\ref{thm:Hausdorff} does not effectivize.
The reason is that
although the  topological closure of a $\Delta^0_2$ set
is closed,  hence $\tbPi^0_1$, it may not be $\Pi^0_1$.
For instance, let $X$ be any countable $\Delta^0_2(\N)$ subset of $\N$.
In the real line, $X$ is $\Delta^0_2(\R)$ and closed
hence $\tbPi^0_1(\R)$, but $X$ is not $\Pi^0_1(\R)$.
Idem in the Baire space with the set $\{f\in\Baire\mid f(0)\in X\}$.
This difficulty is mentioned in \cite{selivanov2005}, page 53, lines 6-7,
with open sets: in the proof of his Theorem 3.10 page 50,
open sets $A_\beta$'s are defined using some $\Delta^0_2$
set $A$: this is  a stumbling block to get $\Sigma^0_1$ sets.
For the Baire space $\Baire$,
Selivanov (2003)  uses a proof different from
Hausdorff's original one.
We adapt it to get the following
weak effective version of Hausdorff's theorem.

\begin{theorem}\label{thm:Hausdorff effective}
Let $(E,(O_n)_{n\in\N},\ll)$ be an effective topological
approximation space.
Then $(F_\sigma\cap G_\delta)(E)
\subseteq \bigcup_{\alpha<\CK}\tlD_\alpha(E)
\subseteq \Delta^0_2(D)$.
In particular, if $(F_\sigma\cap G_\delta)(E)=\Delta_2^0(E)$ then
the effective Hausdorff's theorem holds.
\end{theorem}

\begin{proof}
Without loss of generality, we suppose that the family
$(O_n)_{n\in\N}$ is effectively closed under finite union:
$O_{i_1}\cup\ldots\cup O_{i_k}=O_{\lambda(\{i_1,\ldots,i_k\})}$
for some computable function $\lambda:\FIN\setminus\{\emptyset\}\to\N$.
Using Remark~\ref{rk:Borel effective},
for $\varepsilon=0,1$, let $(I_n^\varepsilon)_{n\in\N}$
be a family of subsets of $\N$ such that
$\{(n,i)\mid i\in I_n^\varepsilon\}$
is computably enumerable and
$$
(*)\qquad
A=\bigcap_{n\in\N} \bigcup_{i\in I_n^1} O_i
\quad,\quad
E\setminus A=\bigcap_{n\in\N} \bigcup_{i\in I_n^0} O_i\ .
$$
We can suppose that the $I_n^\varepsilon$'s are closed
under the above function $\lambda$.
Let $R$ be the computably enumerable set
$R=\{(i,j)\mid O_i\ll O_j\}$.
Let $R^{(t)}$, $I_n^{\varepsilon,t}$ be the finite parts of $R$
and $I_n^t$ obtained after $t$ steps of enumeration.
Let $F_\varepsilon:\N^2\to\N$
be the function such that 
\begin{eqnarray*}
F_\varepsilon(m,t) &=& \max\{p \mid 0\leq p\leq t\textit{\ and\ }
\forall q<p\ (m,\lambda(I_q^{\varepsilon,t}))\in R^{(t)}\}
\end{eqnarray*}
In particular, 
$$(\dagger)\qquad 
O_m
\ \subseteq\
\bigcap_{q<F_\varepsilon(m,t)} \bigcup_{i\in I_q^{\varepsilon,t}} O_i
\ \subseteq\
\bigcap_{q<F_\varepsilon(m,t)} \bigcup_{i\in I_q^{\varepsilon}} O_i.
$$
We define a family $\+T$ of finite sequences of integers in $\N$.
The empty sequence is in $\+T$.
A sequence $(t_0,\ldots,t_k)$ is in $\+T$ if and only if
the following conditions are satisfied:
\begin{enumerate}
\item[(a)]\
$m_0<\ldots<m_k$ and $t_0<\ldots<t_k$,
\item[(b)]\
$(m_{\ell+1},m_\ell)\in R^{(t_{\ell+1})}$ for all $\ell<k$.
In particular, this says that the sequence of subsets
$(O_{m_\ell})_{\ell=0,\ldots,k}$ is decreasing.
\item[(c)]\
For all $\ell\leq k$, $F_0(m_\ell,t_\ell) \neq F_1(m_\ell,t_\ell)$.
For all $\ell<k$,
if $F_\varepsilon(m_\ell,t_\ell) <F_{1-\varepsilon}(m_\ell,t_\ell)$
then $F_{1-\varepsilon}(m_{\ell+1},t_{\ell+1})
              <F_\varepsilon(m_{\ell+1},t_{\ell+1})$.
\end{enumerate}
It is clear that, as a family of finite sequences, $\+T$ is a tree
and is computable.

{\it Claim 1.} The tree $\+T$ is well-founded
(i.e. it has no infinite branch).

{\it Proof of Claim 1.}
Else consider an infinite branch $(m_\ell,t_\ell)_{\ell\in\N}$.
Condition (3) of the Definition \ref{def:effBorelfinite}
and condition (b) of the definition
of $\+T$ imply that the set $\bigcap_{\ell\in\N} O_{m_\ell}$
contains at least one element $x$.
Observe that condition (c) implies that
$F_\varepsilon(m_\ell,t)\geq\lfloor \ell/2\rfloor$.
Using $(\dagger)$ for even and odd $n$'s,
we see that
$x\in \bigcap_{q<\lfloor \ell/2\rfloor} \bigcup_{i\in I_q^{\varepsilon}} O_i$
for all $\ell$ and for $\varepsilon=0,1$.
Thus, $x$ is in both $A$ and $E\setminus A$, a contradiction.
Let us say that a pair $(m,t)\in X$ has type $\varepsilon$ if
$F_\varepsilon(m,t) >F_{1-\varepsilon}(m,t)$.

{\it Claim 2.}  Suppose
$x\in A$ and $(m,t)\in X$ has type $0$ and $x\in O_m$.
Then there exists $(p,u)\in X$ such that
$(p,u)$ has type $1$ and $x\in O_p$ and $(p,m)\in R$, $m<p$ and $t\leq u$.
Switching types $0$ and $1$, the same is true with $x\notin A$.

{\it Proof of Claim 2.}
We treat the sole case $x\in A$,
the other one being trivial modification.
Since $x\notin E\setminus A$,
condition $(*)$ insures that there exists $n$ such that
$x\notin\bigcup_{i\in I_n^0} O_i$.
Since $x\in A$, using again $(*)$, we see that
$x\in\bigcap_{r\leq n}\bigcup_{i\in I_r^1} O_i$.
Choose $i_1,\ldots,i_n$, $p$ such that
$x\in O_{i_r}$ and $(i_r,p)\in R$  for all $r\leq n$.
Let $u>t,n,i_1,\ldots,i_n$ .
Then $F_0(p,u)<n$ whereas $F_1(p,u)\geq n$
hence $(p,u)$ has type $1$ and satisfies Claim 2.
We extend \`a la Kleene-Brouwer the computable well-founded
reverse prefix partial ordering on $\+T$
into a computable total well-ordering $\preceq$ on $\+T$~:
let $\sigma,\tau\in\+T$,

\begin{itemize}
\item
If the sequences $\sigma$ and $\tau$ are prefix comparable,
then we $\preceq$-compare them
relative to the {\it reverse} prefix partial order.
\item
If the sequences $\sigma$ and $\tau$ are not prefix comparable,
then we compare the first elements on which they differ
relative to the usual order on $\N$.
\end{itemize}
Using Ershov's trick, cf. Remark~\ref{rk:omega1CK},
we consider the set
$\+S=\+T\times(\omega+2)$
well-ordered lexicographically using $\preceq$ on $\+T$
and the ordering on the ordinal
$\omega+2=\{0,1,2,\ldots,\omega,\omega+1\}$.
In other words, $\+S$ is obtained from $\+T$ by replacing
each element of $\+T$ by a chain of $\omega+2$ copies of
that element.
We also denote by $\preceq$ the well-ordering on $\+S$.
The following Claim is straightforward.

{\it Claim 3.}  $(\+S,\preceq)$ is a computable well-ordering
and its order type is an even ordinal.
Moreover, the parity of the rank relative to $\preceq$ of an element
$(\sigma,\gamma)\in\+S$ is equal to the parity of the ordinal
$\gamma$. In particular, this parity function is computable.
Attach to any element $(\sigma,\theta)\in\+S$ an open set
$U(\sigma,\theta)$ as follows:
If $\sigma$ is the empty sequence
then $U(\sigma,\theta)=\emptyset$.
Else, if the last pair $(m_k,t_k)$ of $\sigma$ 
has type $\varepsilon\in\{0,1\}$
then
$$
\forall n<\omega\ U(\sigma,n)=\emptyset\quad,\quad
U(\sigma,\omega+\varepsilon)=O_{m_k}\quad,\quad
U(\sigma,\omega+(1-\varepsilon))=\emptyset\ .
$$
If $\alpha$ is the rank of the element $(\sigma,\theta)$ in $T$
then let $A_\alpha=U(\sigma,\theta)$.
The family $(A_\alpha)_{\alpha<\xi}$ is an effective family of
open sets in the sense of Definition~\ref{def:eff Hausdorff}.
Let $\xi$ be the (limit) ordinal to which $(\+S,\preceq)$ is isomorphic.
The following Claim finishes the proof.

{\it Claim 4.} $D_\xi((A_\alpha)_{\alpha<\xi}) = A$.

{\it Proof of Claim 4.}
Suppose $x\in A$.
By Claim 2 there exists a pair $(m,t)$ such that 
the one element sequence $\tau=((m,t))$ is in $\+T$, has type $1$ and
$x\in U(\tau,\omega+1)$.
Thus, we can consider the least ordinal $\alpha$
which is the rank of some $(\sigma,\gamma)\in\+S$
such that $x\in U(\sigma,\gamma)$.
First, we show that $\alpha$ is odd.
Since $U(\sigma,\gamma)$ is not empty, its rank is of the form
$(\omega+2)\delta+\omega+\varepsilon$ where $\varepsilon$ 
is the type of $\sigma$.
If $\varepsilon=0$ then Claim 2 would allow to extend $\sigma$
to $\sigma\ (p,u)\in \+T$ such that $x\in O_p$
and $\sigma\ (p,u)$ has type $1$.
Since $\sigma\ (p,u)$ extends $\sigma$ it has lesser rank in $\+T$.
Thus, $x\in U(\sigma\ (p,u),\delta)$ where $(\sigma\ (p,u),\delta)$
has lesser rank than $(\sigma,\gamma)$, a contradiction.
By definition of $\alpha$, we have
$A_\alpha=U(\sigma,\gamma)$ hence $x\in A_\alpha$.
The choice of $\alpha$ insures that $x$ is in no $A_\beta$
for $\beta<\alpha$.
Thus, $x\in A_\alpha\setminus\bigcup_{\beta<\alpha}A_\beta$
with $\alpha$ odd. This shows that $x\in \tlD((A_\mu)_{\mu<\xi})$.
A similar argument shows that if $x\notin A$ then
$x\in A_\alpha\setminus\bigcup_{\beta<\alpha}A_\beta$
with $\alpha$ even hence $x\notin\tlD((A_\mu)_{\mu<\xi})$.
Summing up, we see that $A=\tlD((A_\mu)_{\mu<\xi})$.
\end{proof}

\begin{corollary}\label{cor:Hausdorff in omega continuous}
In every effective $\omega$-continuous domain $D$,
$$\displaystyle{G_\delta(D)\cap F_\sigma(D)
\subseteq
\bigcup_{\alpha<\aleph_1}\tlD_\alpha(D)
\subseteq \Delta^0_2(D)}.$$
\end{corollary}



\begin{thebibliography}{}

\bibitem[Abramsky \& Jung 1994]{abramsky1994}
Samson~Abramsky and Achim~Jung.
\newblock {\it Domain theory}.
\newblock In Handbook of Logic in Computer Science, vol. III.
Abramsky, Gabbay \& Maibaum, editors.
Oxford University Press, 1994

\bibitem[Becher \& Grigorieff 2012b]{mscsb}
Ver\'onica~Becher and Serge~Grigorieff.
\newblock Wadge hardness in Scott spaces and its effectivization.
\newblock {\it This volume}, 2012.

\bibitem[Bennett \& Lutzer 2009]{bennetlutzer2009}
Harold~Bennett and David Lutzer.
\newblock Strong completeness properties in topology.
\newblock {\it Questions and Answers in General Topology},
27(2009), 107--12, 2009.

\bibitem[de Brecht 2011]{deBrecht2011}
Matthew~de~Brecht.
\newblock Quasi-Polish spaces.
\newblock {\it ArXiv:1108.1445v1}, 40 pages, 2011.

\bibitem[Choquet 1969]{choquet1969}
Gustave Choquet.
\newblock {\it Lectures on analysis. Vol. I: Integration and
topological vector spaces}.
\newblock Benjamin, 1969.

\bibitem[Debs 1984]{debs1984}
G.~Debs.
\newblock{An example of an $\alpha$-favourable topological space
with no $\alpha$-winning tactic.}
\newblock {\it S\'eminaire d'Initiation a l'Analyse
(Choquet-Rogalski-Saint Raymond)}. 1984.

\bibitem[Debs 1985]{debs1985}
G.~Debs.
\newblock{Strat\'egies gagnantes dans certains jeux topologiques.}
\newblock {\it Fundamenta Mathematicae}, 126:93--105, 1985.

\bibitem[Dehornoy 1986]{dehornoy1986}
Patrick~Dehornoy.
\newblock Turing complexity of the ordinals.
\newblock {\it Information Processing Letters}, 23(4):167--170, 1986.

\bibitem[Dorais \& Mummert 2010]{doraismummert2010}
Fran\c{c}ois~G.~Dorais and Carl~Mummert.
\newblock Stationary and convergent strategies
in Choquet games.
\newblock {\it Fundamenta Mathematicae}, 209:59--79, 2010.

\bibitem[Edalat 1997]{edalat1997}
Abbas~Edalat.
\newblock Domains for computation in mathematics, physics and exact real
  aritmetic.
\newblock {\it Bulletin of Symbolic Logic}, 3(4):401--452, 1997.

\bibitem[Ershov 1968]{ershov1968}
Yuri~Ershov.
\newblock On a hierarchy of sets II.
\newblock {\it Algebra and Logic}, 7(4):15--47, 1968.

\bibitem[Galvin \& Telg\'arky 1986]{galvin1986}
Fred~Galvin \& RatislavTelg\'arky.
\newblock Stationary strategies in topological games.
\newblock {\it Topology and its applications}, 22:51--69, 1986.

\bibitem[Gierz \& al. 2003]{BookDomains}
G.~Gierz, K.~H.~.Hofmann, K.~Keimel, J.~D.~Lawson, M.~Mislove
and D.~S.~Scott.
\newblock {\it Continuous Lattices and Domains}.
\newblock Cambridge University Press, 2003.

\bibitem[Grigorieff 1990]{grigorieff1990}
Serge~Grigorieff.
\newblock Every recursive linear ordering has an isomorphic copy in
DTIME-SPACE$(n,\log(n))$.
\newblock {\it Journal of Symbolic Logic}, 55(1):260--276, 1990.

\bibitem[Hertling 1996a]{hertling1996a}
Peter~Hertling.
\newblock Unstetigkeitgrade con Funktionen in der effektiven Analysis.
\newblock {\it PhD thesis, FernUniversity in Hagen}, 1996.

\bibitem[Hertling 1996b]{hertling1996b}
Peter~Hertling.
\newblock Topological complexity with continuous operations.
\newblock {\it Journal of Complexity}, 12(4):315--338, 1996.

\bibitem[Kechris 1995]{kechrisBook}
Alexander~S.~Kechris.
\newblock {\it Classical Descriptive Set Theory}.
\newblock Springer, 1995.

\bibitem[K\"unzi 1983]{kunzi1983}
Hans-Peter~K\"unzi.
\newblock On strongly quasi-metrizable spaces.
\newblock {\it Archiv der Mathematik}, 41(1):57--63, 1983.

\bibitem[Kuratowski 1966]{kuratowskiBook}
Kazimierz Kuratowski.
\newblock {\it Topology. Volume I}.
\newblock Academic Press, 1966.

\bibitem[Marker 2002]{marker2002}
David Marker.
\newblock {\it Descriptive set theory}.
\newblock Lecture Notes.
On Marker's home page, 2002.

\bibitem[Moschovakis 1979/2009]{moschovakisbook}
Yiannis Moschovakis.
\newblock {\it Descriptive set theory}, volume 155.
\newblock American Mathematical Society.
First edition 1979, second edition 2009.

\bibitem[Oxtoby 1957]{oxtoby1957}
J.C.~Oxtoby.
\newblock {\it The Banach-Mazur game and Banach Category Theorem}.
\newblock In Contributions to the theory of games,
Vol. III, Annals of Math. Studies 39:159--163, 1957.

\bibitem[Schmidt 1966]{schmidt1966}
W.W.~Schmidt.
\newblock On badly approximable numbers and certain games.
\newblock {\it Transactions Amer. Math. Soc.},
123:178--199, 1966.

\bibitem[Selivanov 2003]{selivanov2003}
Victor~L. Selivanov.
\newblock Wadge degrees of $\omega$-languages of deterministic Turing machines.
\newblock {\it Theoretical Informatics and Applications}, 37(1):67--83, 2003.
\newblock Extended abstract in STACS 2003 Proceedings,
Lecture Notes in Computer Science 2607:97--108, 2003. 

\bibitem[Selivanov 2005]{selivanov2005}
Victor~L. Selivanov.
\newblock Hierarchies in $\varphi$-spaces and applications.
\newblock {\it Mathematical Logic Quaterly}, 51(1):45--61, 2005.

\bibitem[Selivanov 2006]{selivanov2006}
Victor~L. Selivanov.
\newblock Towards a descriptive set theory for domain-like structures.
\newblock {\it Theoretical Computer Science}, 365(3):258--282, 2006.

\bibitem[Selivanov 2008]{selivanov2008}
Victor~L. Selivanov.
\newblock On the difference hierarchy in countably based $T_0$-spaces.
\newblock {\it Electronic Notes in Theoretical Computer Science}, 221:257--269, 2008.

\bibitem[Spector 1955]{spector1955}
Clifford ~Spector.
\newblock Recursive well-orderings.
\newblock {\it Journal of Symbolic Logic}, 20(2):151--163, 1955.

\bibitem[Tang 1981]{tang1981}
A.~Tang.
\newblock Wadge reducibility and hausdorff difference hierarchy in {$P\omega$}.
\newblock Lectures Notes in Mathematics 871:360--371, 1981.

\bibitem[Weihrauch 2000]{weihrauchbook2000}
Klaus~Weihrauch.
\newblock {\it Computable analysis. An introduction}.
\newblock Springer, 2000.

\end{thebibliography}
\end{document}